\documentclass[imslayout,preprint]{imsart} 
\def\journal@name{} 
\usepackage[margin=1.5in]{geometry}
\usepackage{amssymb,amsthm,mathtools}
\usepackage{mathtools}
\usepackage[unicode,colorlinks,citecolor=blue,urlcolor=black,linkcolor=blue,pdfborder={0 0 0}]{hyperref}
\usepackage[capitalize,sort,compress]{cleveref} %
\usepackage{comment}
\usepackage{graphicx}
\usepackage{enumitem}
\usepackage{framed}
\usepackage{dcolumn}
\usepackage{chngcntr}

\usepackage{url} %
\usepackage[usenames]{color} %
\usepackage{subfig}
\usepackage[sort]{natbib}
\usepackage{enumitem}

\usepackage{datetime}
\usepackage[normalem]{ulem}
\usepackage[linesnumbered,ruled,vlined]{algorithm2e}

\usepackage{microtype} %
\usepackage{booktabs} %
\usepackage{graphicx}

\newtheorem{assumption}{Assumption}
\newtheorem{theorem}{Theorem}
\newtheorem{proposition}{Proposition}
\newtheorem{lemma}{Lemma}
\crefname{theorem}{Theorem}{Theorems}
\crefname{lemma}{Lemma}{Lemmas}
\crefname{proposition}{Proposition}{Propositions}
\crefname{corollary}{Corollary}{Corollaries}
\crefname{conjecture}{Conjecture}{Conjectures}
\crefname{definition}{Definition}{Definitions}
\crefname{assumption}{Assumption}{Assumptions}
\crefname{example}{Example}{Examples}
\crefname{remark}{Remark}{Remarks}

\newcommand{\charalphabetmacro}[3]{
	\def\mydeffoo##1{\expandafter\def\csname #1##1\endcsname{#2{##1}}}
	\def\mydefallfoo##1{\ifx##1\mydefallfoo\else\mydeffoo##1\expandafter\mydefallfoo\fi}
	\expandafter \mydefallfoo #3\mydefallfoo
}

\newcommand{\texalphabetmacro}[3]{
	\def\mydeftex##1{\expandafter\def\csname #1##1\endcsname{#2{\csname ##1\endcsname}}}
	\def\mydefalltex##1{\ifx\mydefalltex##1\else\mydeftex{##1}%
		\lowercase{\mydeftex{##1}}\expandafter\mydefalltex\fi}
	\expandafter \mydefalltex #3\mydefalltex
}

\newcommand{\upperCaseRomanLetters}{ABCDEFGHIJKLMNOPQRSTUVWXYZ}
\newcommand{\lowerCaseRomanLetters}{abcdefghijklmnopqrstuvwxyz}
\newcommand{\lowerCaseRomanLettersNoMHT}{abcdefgijklnopqrsuvwxyz}

\newcommand{\lowerCaseRomanLettersNoMF}{abcdeghijklnopqrstuvwxyz}

\newcommand{\lowerCaseGreekLetters}{{alpha}{beta}{gamma}{delta}{epsilon}{zeta}{eta}{theta}{iota}{kappa}{lambda}{mu}{nu}{xi}{omicron}{pi}{rho}{sigma}{tau}{upsilon}{phi}{chi}{psi}{omega}}
\newcommand{\lowerCaseGreekLettersNoEta}{{alpha}{beta}{gamma}{delta}{epsilon}{zeta}{theta}{iota}{kappa}{lambda}{mu}{nu}{xi}{omicron}{pi}{rho}{sigma}{tau}{upsilon}{phi}{chi}{psi}{omega}}
\newcommand{\upperCaseGreekLettersInLaTeX}{{Gamma}{Delta}{Theta}{Lambda}{Xi}{Pi}{Sigma}{Upsilon}{Phi}{Psi}{Omega}}

\charalphabetmacro{bar}{\bar}{\upperCaseRomanLetters}
\charalphabetmacro{bar}{\bar}{\lowerCaseRomanLetters}

\charalphabetmacro{b}{\boldorbar}{\upperCaseRomanLetters}
\charalphabetmacro{b}{\boldorbar}{\lowerCaseRomanLettersNoMF} % avoid clash with \bm and \bf
\texalphabetmacro{b}{\boldorbar}{\lowerCaseGreekLettersNoEta} % avoid \beta <-> \bar{\eta} clash

\texalphabetmacro{b}{\boldorbar}{\upperCaseGreekLettersInLaTeX}

\charalphabetmacro{mc}{\mathcal}{\upperCaseRomanLetters}

\newcommand{\hatmathcal}[1]{\hat{\mathcal{#1}}}
\charalphabetmacro{hmc}{\hatmathcal}{\upperCaseRomanLetters}

\charalphabetmacro{h}{\hat}{\upperCaseRomanLetters}
\charalphabetmacro{h}{\hat}{\lowerCaseRomanLettersNoMHT}
\texalphabetmacro{h}{\hat}{\lowerCaseGreekLetters}
\texalphabetmacro{h}{\hat}{\upperCaseGreekLettersInLaTeX}

\newcommand{\boldhat}[1]{\mathbf{\hat{#1}}}
\charalphabetmacro{bh}{\boldhat}{\upperCaseRomanLetters}
\charalphabetmacro{bh}{\boldhat}{\lowerCaseRomanLetters}
\texalphabetmacro{bh}{\boldhat}{\lowerCaseGreekLetters}
\texalphabetmacro{bh}{\boldhat}{\upperCaseGreekLettersInLaTeX}

\charalphabetmacro{wh}{\widehat}{\upperCaseRomanLetters}
\charalphabetmacro{wh}{\widehat}{\lowerCaseRomanLettersNoMHT}

\charalphabetmacro{td}{\tilde}{\upperCaseRomanLetters}
\charalphabetmacro{td}{\tilde}{\lowerCaseRomanLetters}
\texalphabetmacro{td}{\tilde}{\lowerCaseGreekLetters}
\texalphabetmacro{td}{\tilde}{\upperCaseGreekLettersInLaTeX}

\newcommand{\datarv}{\mathbf{X}}

\newcommand{\opt}{\star}

\newcommand{\mle}[1]{\hat\param}

\newcommand{\spc}{\text{SPC}}  % SPC index

\newcommand{\spcrepobsobs}[1]{\datarv^{\spc}_{\text{obs}}}

\newcommand{\kl}[2]{\mathrm{KL}(#1 \mid #2)}

\newcommand{\klest}[2]{\mathrm{\widehat{KL}}(#1 \mid #2)}
\newcommand{\klestsub}[4]{\mathrm{\widehat{KL}^{#1}}_{#2}(#3 \mid #4)}

\newcommand{\data}[1]{x_{#1}}
\newcommand{\numobs}{N}
\newcommand{\numcomps}{K}

\newcommand{\allparam}{\theta}
\newcommand{\param}{\phi}

\newcommand{\discr}[2]{\mcD(#1\mid #2)}
\newcommand{\discrest}[2]{\hat\mcD(#1\mid #2)}

\newcommand{\blmetric}{d_\mathrm{BL}}

\renewcommand{\Pr}{\text{pr}} 
\newcommand{\dee}{\mathrm{d}}
\newcommand{\op}{o_{P}}
\newcommand{\BLnorm}[1]{\Vert #1 \Vert_{\mathrm{BL}}}

\def\argmin{\operatornamewithlimits{arg\,min}}
\newcommand{\distNamed}[1]{{\sf{#1}}}
\newcommand{\distCat}{\distNamed{Categorical}}
\newcommand{\distNBinom}{\distNamed{NegBin}}
\newcommand{\distPoiss}{\distNamed{Poiss}}
\newcommand{\distNorm}{\mathcal{N}}
\newcommand{\distSNorm}{\mathcal{SN}}

\graphicspath{{./figures/}}

\begin{document}

\begin{frontmatter}
\title{Structurally Aware Robust \\ Model Selection for Mixtures}
\runtitle{~Structurally Aware Robust Model Selection for Mixtures}
\runauthor{J.\ Li and J.\ H.\ Huggins~}

\begin{aug}
	\author[A]{\fnms{Jiawei} \snm{Li}\ead[label=e2]{jwli@bu.edu}}
	\and
	\author[A]{\fnms{Jonathan H.} \snm{Huggins}\ead[label=e1]{huggins@bu.edu}}
	\address[A]{Department of Mathematics \& Statistics, Boston University, \printead{e2,e1}}
\end{aug}

\maketitle

\begin{abstract}
Mixture models are often used to identify meaningful subpopulations (i.e., clusters) in observed data such that the subpopulations have a real-world interpretation (e.g., as cell types).
However, when used for subpopulation discovery, mixture model inference is usually ill-defined \emph{a priori} because the assumed observation model is only an approximation to the true data-generating process. 
Thus, as the number of observations increases, rather than obtaining better inferences, the opposite occurs: 
the data is explained by adding spurious subpopulations that compensate for the shortcomings of the observation model. 
However, there are two important sources of prior knowledge that we can exploit to obtain well-defined results no matter the dataset size: known causal structure (e.g., knowing that the latent subpopulations cause the observed signal but not vice-versa) and a rough sense of how wrong the observation model is (e.g., based on small amounts of expert-labeled data or some understanding of the data-generating process). 
We propose a new model selection criteria that, while model-based, uses this available knowledge to obtain mixture model inferences that are robust to misspecification of the observation model. %
We provide theoretical support for our approach by proving a first-of-its-kind consistency result under intuitive assumptions.
Simulation studies and an application to flow cytometry data demonstrate our model selection criteria consistently finds the correct number of subpopulations. 
\end{abstract}

\end{frontmatter}

\section{Introduction}
\label{sec:intro}
In scientific applications, mixture models are often used to discover unobserved subpopulations or distinct types that generated the observed data. 
Examples include cell type identification (e.g., using single-cell assays) \citep{Gorsky:2020,Prabhakaran:2016}, behavioral genotype discovery (e.g., using gene expression data) \citep{Stevens:2019}, and psychology (e.g., patterns of IQ development) \citep{Bauer:2007}.
Further examples include disease recognition \citep{Greenspan:2006,Adelino:2007,Ghorbani:2016}, anomaly detection  \citep{Zong:2018}, 
and types of abnormal heart rhythms from ECG \citep{Ghorbani:2016}).
Since the true number of types or subpopulations $K_{\circ}$ is usually unknown \emph{a priori}, a key challenge in these applications is to determine $K_{\circ}$ \citep{Cai:2021,Guha:2021,Fruhwurth:2006,Miller:2019}. 
While numerous inference methods provide consistent estimation when the mixture components are well-specified,
if the components are misspecified, then standard methods do not work as intended. 
As the number of observations increases, rather than obtaining better inferences, the opposite occurs:
the data is explained by adding spurious latent structures that compensate for the shortcomings of the observation model, which is known as \emph{overfitting} \citep{Cai:2021,Miller:2019,Fruhwurth:2006}.

To address the overfitting problem in the misspecified setting, various robust clustering methods have been developed. 
Heavy-tailed mixtures  \citep{Archambeau:2007,Christopher:2004,Wang:2018} and sample re-weighting \citep{Forero:2011} 
both aim to account for slight model mismatches and outlier effects. 
However, these methods require choosing the number of clusters for the mixture model beforehand,
and their performance heavily relies knowing the number of subpopulations.
Most closely related to our approach, \citet{Miller:2019} provides a novel perspective on establishing robustness without prior knowledge of number of clusters. 
They employ a technique they call \emph{coarsening}: 
instead of assuming the data was generated from the assumed model, the \emph{coarsened posterior} allows a certain degree of divergence between the model and data distribution. 
This flexibility permits the overall mixture model to deviate from the true underlying data by a set threshold.
While this approach shows good robustness properties in practice, imposing the robustness threshold on the \emph{overall} mixture density does not guarantee the convergence to a 
meaningful number of components.
Moreover, their approach to determining $\numcomps_{o}$ involves post-processing, specifically eliminating the minimal clusters.
An additional challenge when using the coarsened posterior approach is its high computational cost:
it requires running Markov chain Monte Carlo dozens of times to heuristically determine a suitable robustness threshold.
In this paper, we propose a model selection criterion to address the statistical and computational shortcomings of existing approaches. 
In contrast to the coarsened posterior approach, which is based on the overall degree of model--data mismatch, 
our approach operates at the \emph{component level}. %
Moreover, our criterion is more flexible than previous approaches and can intuitively incorporate expert and prior knowledge about the nature of the model misspecification. 
One source of flexibility is the ability to combine it with any parameter estimation technique for mixture models.
Another is that it supports a wide range of discrepancy measures between distributions, which can be chosen in an application-specific manner. 
Given a fitted model for each candidate number of components, the computational cost required to compute our criterion is determined by the cost of 
estimating the discrepancy measure, which is typically quite small. 
Because we employ component-wise robustness, we are able to show that, under natural conditions and for a wide variety of discrepancies, 
including the Kullback--Leibler divergence and the maximum mean discrepancy,
our method identifies the number of components correctly in large datasets.
We validate the effectiveness and computational efficiency of our model selection criterion when using the Kullback--Leibler divergence as the discrepancy measure
through a combination of simulation studies and an application to cell type discovery using flow cytometry data.

\section{Setting and Motivation}
\label{sec:motivation}

Consider data $\data{1:\numobs}= (\data{1}, \dots, \data{\numobs}) \in \mcX^{\otimes \numobs}$ independently and identically distributed from some unknown distribution $P_o$ defined on the measurable space $(\mcX, \mcB_x)$. 
When discovering latent types, we assume that $P_o = \sum_{k=1}^{\numcomps_o} \pi_{ok} P_{ok}$, 
where $\numcomps_o$ is the number of types, $P_{ok}$ is the distribution of observations from the $k$th type, and 
$ \pi_{ok}$ is the probability that an observation belongs to the $k$th type.
Thus, $\pi_{o} = (\pi_{o1},\dots,\pi_{o\numcomps_{o}}) \in \Delta_{\numcomps_o} = \{ \pi \in \mathbb{R}_{+}^{\numcomps_{o}} \mid \sum_{k=1}^{\numcomps_{o}} \pi_{k} = 1 \}$, the $(\numcomps_o-1)$-dimensional probability simplex, and we assume $\pi_{ok} > 0$ so all components contribute observations. 
We posit a mixture model parameterized by $\theta \in \Theta = \bigcup_{\numcomps=1}^{\infty} \Theta^{(\numcomps)}$,
where $\Theta^{(\numcomps)} = \Delta_{\numcomps} \times \Phi^{\otimes \numcomps}$ for measurable space $(\Phi, \mcB_{\Phi})$
and $\phi \in \Phi$  parameterizes the family of mixture component distributions $\mcF = \{ F_{\phi} \mid \phi \in \Phi \}$. 
For $\theta = (\pi, \phi_{1},\dots,\phi_{\numcomps}) \in \Theta$, the mixture model distribution is 
$G_{\theta} = \sum_{k=1}^{\numcomps}\pi_{k}F_{\param_{k}}$. 
Given $\data{1:\numobs}$, our goals are to (1) find $\numcomps_{o}$ and (2) find a parameter 
estimate $\hat\theta = (\hat\pi, \hat\phi_{1},\dots,\hat\phi_{\numcomps_{o}})  \in \Theta^{(\numcomps_{o})}$ such that (possibly after a reordering of the components)
$\hat\pi \approx \pi_{o}$ and $F_{\hat\phi_{k}} \approx P_{ok}$. 

Finding $\numcomps_{o}$ can be challenging for standard model selection approaches given the presence of model misspecification \citep{Cai:2021,Guha:2021,Fruhwurth:2006,Miller:2019}.
Specifically, as the number of observations increases, model selection criteria tend to create additional clusters to compensate for the model--data mismatch and thus overestimates
$\numcomps_{o}$. 
The following toy example illustrates this phenomenon.
Suppose data is generated from a mixture of $\numcomps_{o} = 2$ skew normal distributions and we use a Gaussian mixture to model the data.
Here the level of model--data mismatch is controlled by the skewness parameter of each skew normal component in the true generative distribution $P_{o}$. 
We consider the following scenarios: two equal-sized clusters with the same level of misspecification (denoted \texttt{same})
and two equal-sized clusters with different levels of misspecification (denoted \texttt{different}).  
We compare three methods on their ability to find $\numcomps_{o}$: standard expectation--maximization with the Bayesian information criterion \citep{Chen:1998}, 
the coarsened posterior \citep{Miller:2019}, and expectation--maximization with our proposed model selection criterion. 
See \cref{sec:simulation-gauss} for further details about the experimental set-up. 

As shown in \cref{fig:motivate-comparison}, in both scenarios the Bayesian information criterion selects $\numcomps > \numcomps_{o}$ to capture the skew-normal distributions.
While the coarsened posterior  performs well in the \texttt{same} case,
its limitations become evident when the degree of misspecification differs significantly between components.
In the \texttt{different} scenario, the coarsened posterior overfits the cluster with a larger degree of misspecification.

\begin{figure}[tp]
	\centering	
	\includegraphics[width=65mm]{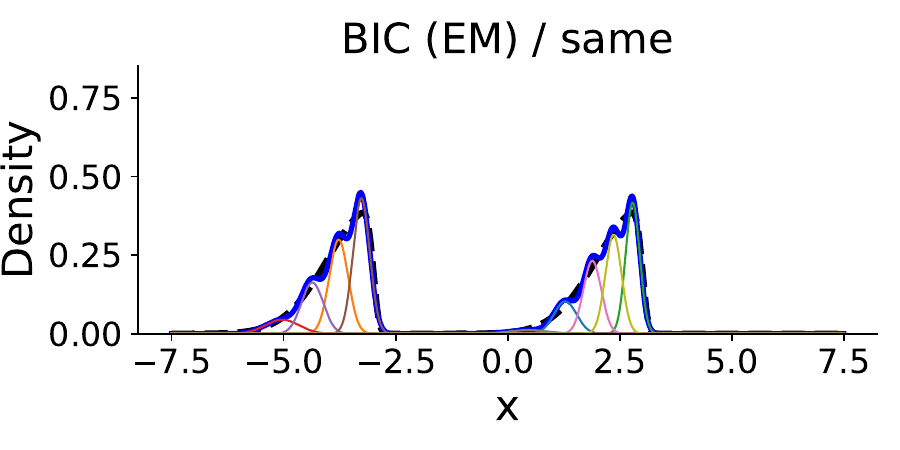}
	\includegraphics[width=65mm]{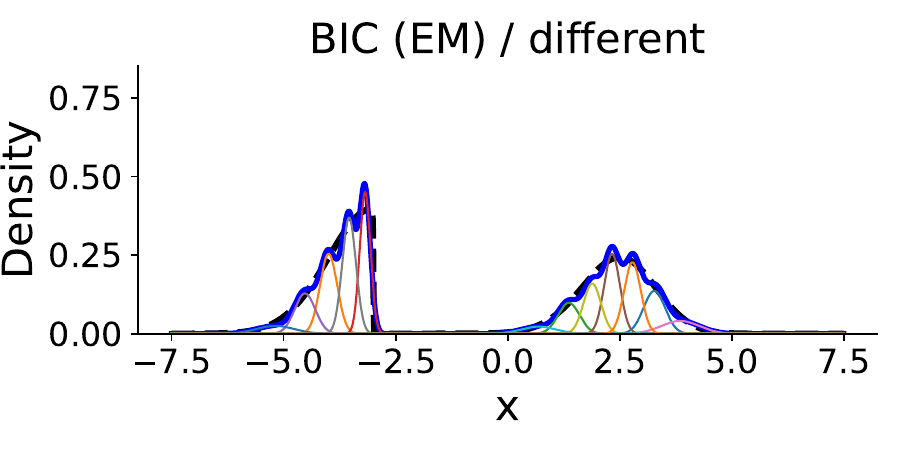}\\
	\includegraphics[width=65mm]{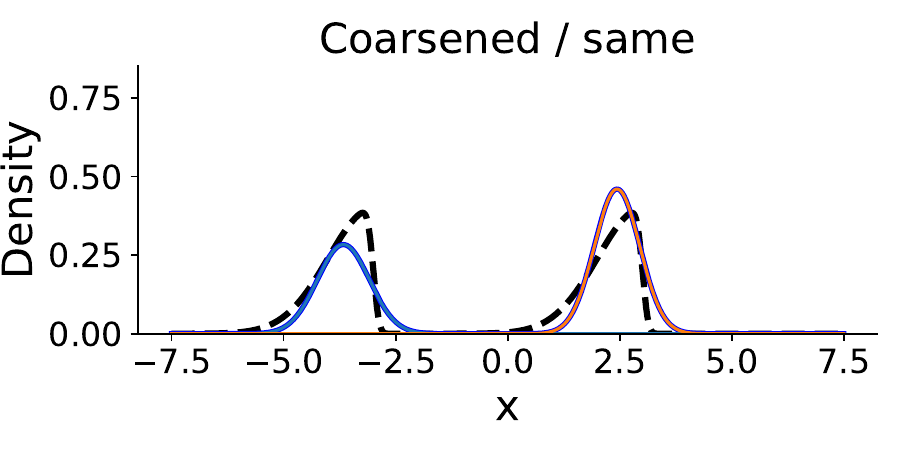}
	\includegraphics[width=65mm]{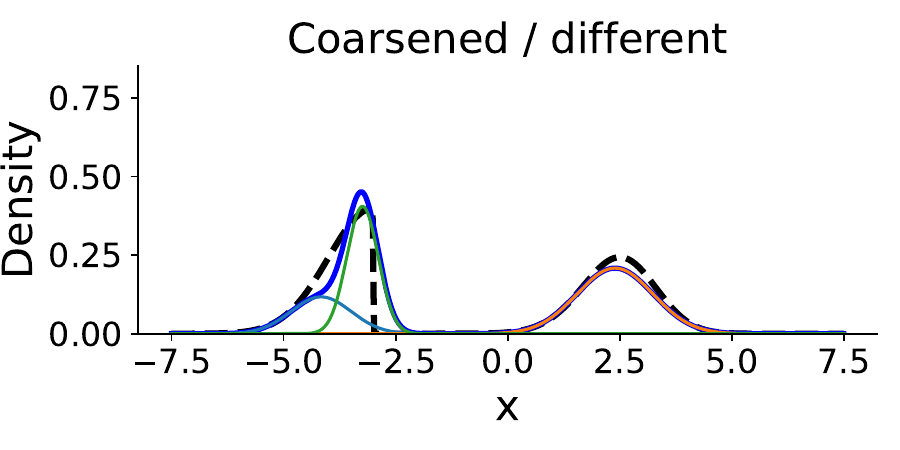}\\
	\includegraphics[width=65mm]{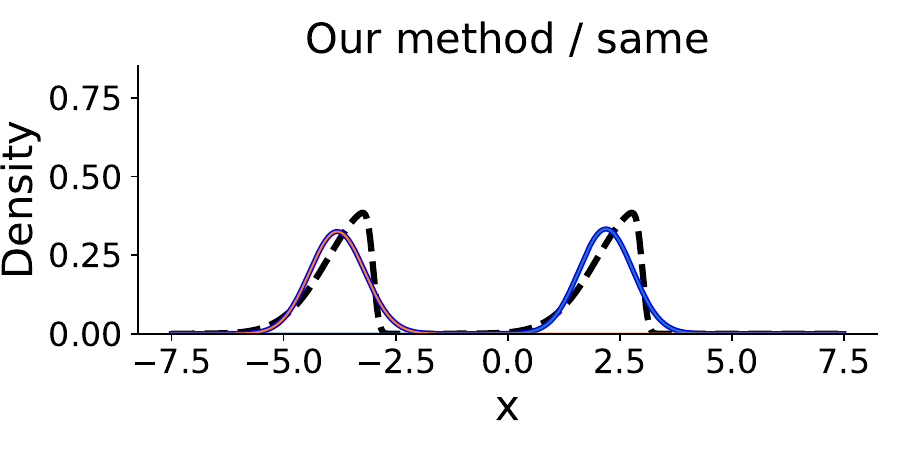}
	\includegraphics[width=65mm]{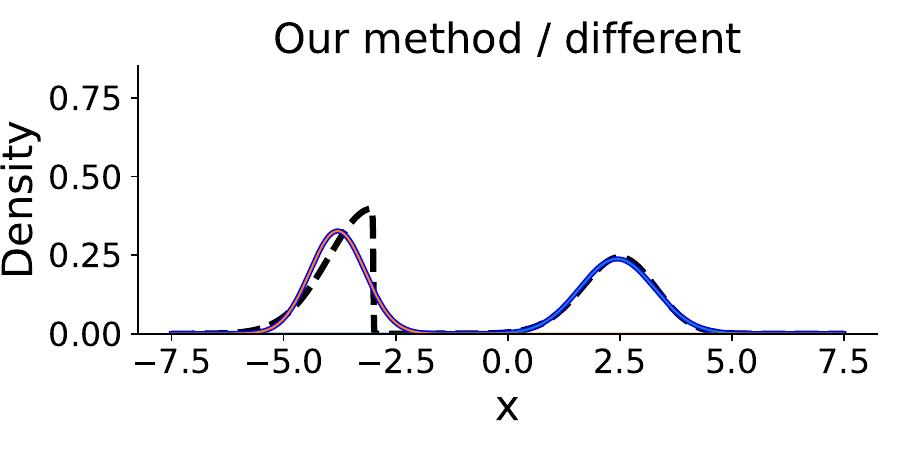}
	\caption{
		For the mixture of skew-normals example from \cref{sec:motivation}, 
		each panel shows the density of $P_{o}$ (dashed lines) and the densities of the fitted Gaussian mixture model and each
		component distribution (solid lines) using $\numobs = 10\,000$ observations. 
		Results are given for three approaches:
		expectation--maximization with the Bayesian information criterion (first row), 
		the coarsened posterior (second row),
		and our robust model selection method (third row).} 
	\label{fig:motivate-comparison}
\end{figure}

Another disadvantage of the coarsening approach comes from computational considerations. 
Letting $\pi_{0}$ denote the prior density and $\alpha > 0$, the coarsened posterior is given by
$$
\pi_{\alpha}(\theta \mid \data{1:\numobs}) \propto \pi_{0}(\theta)\prod_{n=1}^{\numobs} f_{\phi}(\data{n})^{\frac{\alpha}{\numobs +\alpha}},
$$
where $\alpha$ controls the expected degree of misspecification
(roughly speaking, that the Kullback--Leibler divergence between the population and estimated distributions 
is of order $1/\alpha$). 
Selecting $\alpha$ requires a grid search over dozens or more plausible values.
Since a separate power posterior must be estimated for each $\alpha$ value, which typically requires using a slow Markov chain Monte Carlo algorithm,
the coarsened posterior has a high computational cost.

To address the limitations of standard model selection methods and the coarsened posterior, we introduce a computationally efficient and statistically
sound model selection criterion. 
As shown in \cref{fig:motivate-comparison}, our method -- which we describe in detail in the next section -- identifies the correct value for $\numcomps_{o}$ in both scenarios. %

\section{Structurally Aware Model Selection}
\label{sec:method}

\subsection{Method}

Our starting point for developing our robust model selection approach is to rewrite the model in a 
form which captures the known causal structure and makes the component distributions -- which are the source of misspecification -- explicit.  
The posterior, the coarsened posterior, and standard model selection methods based on point estimates are all based (at least asymptotically)
on formulating the mixture model as 
\begin{equation}
\data{n} \sim  G_{\theta} = \sum_{k=1}^{K} \pi_{k}F_{\param_{k}} \quad (n=1,\ldots, N).  \label{eq:marginal-mixture-model}
\end{equation}
However, \cref{eq:marginal-mixture-model} implicitly posits that the structural causal model for each observation is 
$
x_{n} \gets h(\param, \pi; \varepsilon_{n}),
$
where $\param = (\param_{1}, \dots, \param_{\numcomps})$, $\pi = (\pi_{1},\dots, \pi_{\numcomps})$, $h$ is a deterministic function, and $\varepsilon_{n}$ is an
independent ``noise'' random variable. 
Since $h$ is treated as a black box, it is impossible for inference methods based on \cref{eq:marginal-mixture-model} to always
correctly account for component-level misspecification. 
This can lead to deceptive identification regarding the actual model fit from the perspective of individual components, as illustrated in \cref{sec:motivation}. 

An alternative way to formulate the mixture model is in the ``uncollapsed'' (latent variable) form
\begin{equation}
	\begin{aligned}
		z_{n} & \sim \distCat(\pi), \quad & %
		\data{n} \mid z_{n} = k & \sim F_{\param_{k} } && (n=1,\dots,\numobs).
	\end{aligned} \label{eq:mixture-model}
\end{equation}
\Cref{eq:mixture-model} can also be written as the structural causal model
\begin{equation}
	\begin{aligned}
		z_{n} &\gets h_{z}(\pi; \varepsilon_{z,n}), \quad &
		x_{n} &\gets h_{x}(\param, z_{n}; \varepsilon_{x,n}) && (n=1,\dots,\numobs),
		\label{eq:stare-mixture-model}
	\end{aligned}
\end{equation}
where $h_{z}$ and $h_{x}$ are deterministic functions and $\varepsilon_{z,n}$ and $\varepsilon_{x,n}$ are independent ``noise'' random variables. 
For a model selection criterion to reliably select the true number of components, it must make use of the causal structure given in \cref{eq:stare-mixture-model}.
In particular, we want a method that exploits the fact that we know the relationship $x_{n} \gets h_{x}(\param, z_{n}; \varepsilon_{x,n})$ is misspecified. 
Hence, we describe our approach as being ``structurally aware.''
We also rely on the assumption that $z_{n} \gets h_{z}(\pi; \varepsilon_{z,n})$ is well-specified, although it is possible that assumption could be relaxed as well. 

To construct a structurally aware model selection method that is robust to misspecification, we must have a way of measuring how different the estimated 
parametric component distribution $F_{\phi_{k}}$ is from the true component distribution $P_{ok}$.
Therefore, we introduce a divergence $\discr{P_{ok}}{F_{\phi_{k}}}$, which can be chosen by the user based on application-specific considerations.
We assume that for some $\rho > 0$ representing the degree of misspecification, we can estimate component parameters $\phi_{\star k} \in \Phi$ such that 
$\discr{P_{ok}}{F_{\phi_{\star k}}} < \rho$ ($k = 1,\dots, \numcomps_{o}$).
Let $z = (z_1, \ldots, z_{\numobs})$ denote the cluster assignments for the observations defined in \cref{eq:mixture-model}.
The set of observations assigned to component $k$, which we denote $X_{k}(z) = \{ \data{n} \mid z_{n} = k, n =1,\ldots,\numobs\}$, 
provides samples that are approximately distributed according to $P_{ok}$. 
So, we require an estimator $\discrest{X_{k}(z)}{F_{\phi_{k}}}$ which, when the empirical distribution of $X_{k}(z)$ 
converges to some limiting distribution $\tilde{P}_{k}$ as $\numobs \to \infty$, provides a consistent estimate of 
$\discr{\tilde{P}_{k}}{F_{\phi_{k}}}$. 

For clarity, we will often write $\allparam^{(\numcomps)}$, $X_{k}^{(\numcomps)}(z^{(\numcomps)})$, etc.\ to denote that these quantities are
associated with the mixture model with $\numcomps$ components.
Given $\discr{\cdot}{\cdot}$, $\discrest{\cdot}{\cdot}$, and $\rho$, we would like to design a loss function that will be minimized when $K = K_{o}$.
Hence, the loss function should not decrease if all estimated discrepancies are less than $\rho$. 
Based on this requirement, define the \emph{structurally aware loss}
\begin{equation}
	\mcR^{\rho}(\allparam^{(\numcomps)}; z^{(\numcomps)}, \data{1:\numobs}) = \sum_{k = 1}^{\numcomps} |X_{k}^{(\numcomps)}(z^{(\numcomps)})|\max(0, \discrest{X_{k}^{(\numcomps)}(z^{(\numcomps)})}{F_{\param_{k}^{(\numcomps)}}) - \rho}. \label{eq:general-modified-loglik}
\end{equation}
By construction, the structurally aware loss is nonnegative and equal to zero (and therefore minimized) if $\discrest{X_{k}^{(\numcomps)}(z^{(\numcomps)})}{F_{\param_{k}^{(\numcomps)}}} < \rho$ for all $k$.
On the one hand, if $\lim_{\numobs \to \infty}\discrest{X_{k}^{(\numcomps)}(z^{(\numcomps)})}{F_{\param_{k}^{(\numcomps)}}}> \rho$ for some $k$, then the scaling by $|X_{k}^{(\numcomps)}(z^{(\numcomps)})|$ ensures 
the loss tends to $\infty$ as $N \to \infty$. 
Therefore, by minimizing \cref{eq:general-modified-loglik} with respect to $K$, we can exclude $\numcomps < \numcomps_{o}$.

On the other hand, to asymptotically rule out $\numcomps > \numcomps_{o}$, we 
introduce an Akaike information criterion-like penalty term, leading to the \emph{penalized structurally aware loss}
\begin{equation}
	\mathcal{R}^{\rho,\lambda}(\allparam^{(\numcomps)}; z^{(\numcomps)}, \data{1:\numobs}) = \mathcal{R}^{\rho}(\allparam^{(\numcomps)}; z^{(\numcomps)},\data{1:\numobs}) + \lambda \numcomps, \label{eq:pen-general-modified-loglik}
\end{equation}
where $\lambda > 0$ controls the strength of the penalty. %
Given choices for $\lambda$, $\discr{\cdot}{\cdot}$, $\discrest{\cdot}{\cdot}$, and the maximum
number of components to consider,  $\numcomps_{\max}$, our robust, structurally aware model selection procedure is as follows:

\begin{algorithm}[H]
	\SetKwInOut{Input}{Input}
	\SetKwInOut{Output}{Output}

	\For{$\numcomps=1,\ldots, \numcomps_{\max}$ }{
		Obtain parameter estimate $\hat{\allparam}^{(\numcomps)} \in \Theta^{(\numcomps)}$ \\
			Sample $z^{(\numcomps)}$ from $p(z^{(\numcomps)} \mid \hat{\allparam}^{(\numcomps)}; x_{1:\numobs})$
	}
	If necessary, determine an appropriate value for $\rho$ \\
	Compute $\hat{\numcomps} = \argmin_{\numcomps} \mathcal{R}^{\rho,\lambda}(\hat\allparam^{(\numcomps)}; z^{(\numcomps)}, \data{1:\numobs})$
	
	\Return{$\hat{\allparam}^{(\hat{\numcomps)}}$}
	\caption{Robust, structurally aware model selection}
	\label{algo:model-selection}
	
\end{algorithm}

We next turn to rigorously justifying our proposed approach by proving a first-of-its-kind consistency result (\cref{sec:theory}).
We then provide detailed guidance on using our approach in practice (\cref{sec:guidance}).

\section{Model Selection Consistency} %
\label{sec:theory}

\subsection{A General Consistency Result}

We now show that our model selection procedure consistently estimates $K_{o}$ under reasonable assumptions.
First, we consider the requirements for $\discr{\cdot}{\cdot}$ and $\discrest{\cdot}{\cdot}$.
Some discrepancies are finite only under strong conditions. 
For example, for the Kullback--Leibler divergence, $\discr{P}{Q} < \infty$ only if $P$ is absolutely continuous with respect to $Q$.
However, we must work with the discrete estimates of each component, so such absolute continuity conditions may not hold.
Therefore, we introduce a possibly weaker metric $d$ on probability measures that detects empirical convergence. 
\begin{assumption}
	For $y_{\numobs,n} \in \mathcal{X}$ $(\numobs = 1,2,\dots; n = 1,\dots,\numobs)$, 
	define the empirical distribution $\hat{P}_{\numobs} = \numobs^{-1}\sum_{n=1}^{\numobs}\delta_{y_{N,n}}$ 
	and assume $\hat{P}_{\numobs} \rightarrow P$ in distribution. 
	The metric $d$, discrepancy $\mathcal{D}$, and estimator $\hat{\mathcal{D}}$ satisfy the following conditions:
	\begin{enumerate}[label = (\alph*),ref=(\theassumption-\alph*)]
		\item \emph{The metric detects empirical convergence:} $d(\hat{P}_\numobs, P) \to 0$ as $\numobs \to \infty$. 
		\item  \emph{The metric is jointly convex in its arguments:} for all $w \in (0, 1)$ and distributions $P$, $P'$, $Q$, $Q'$, 
		\[
		d(wP + (1 - w)P', wQ + (1 - w)Q') \le w\,d(P, Q) + (1 - w)\,d(P', Q').
		\]
		\item  \emph{The discrepancy estimator is consistent:} For any distributions $P, Q$, 
		if $\discr{P}{Q} < \infty$ and $\hat P_{N} \to P$ in distribution, 
		then $\discrest{\hat{P}_{\numobs}}{Q} \to \discr{P}{Q}$ as $\numobs \to \infty$. 
		\item \emph{Smoothness of the discrepancy estimator:} The map $\param \mapsto \discrest{\hat{P}_{\numobs}}{F_{\param}}$ is continuous. 
		\item  \emph{The discrepancy bounds the metric:} there exists a continuous, 
		non-decreasing function $\phi : \mathbb{R} \to \mathbb{R}$ such that 
		$d(P,Q) \leq \phi(\discr{P}{Q})$ for all distributions $P, Q$. 
		\item \emph{The metric between components is finite:} For all $\param, \param' \in \Theta$, 
		we have $d(F_{\param}, F_{\param'}) < \infty$. 
	\end{enumerate}	
	\label{assump:metric-discr-conditions}
\end{assumption}

A wide variety of metrics satisfy \cref{assump:metric-discr-conditions}(a), including the bounded Lipschitz metric, 
the Kolmogorov metric, maximum mean discrepancies with sufficiently regular bounded kernels, and the Wasserstein metric with a bounded cost function \citep{Vaart:1996,Sriperumbudur:2010,SimonGabriel:2018,Villani:2009}. 
\Cref{assump:metric-discr-conditions}(b) also holds for a range of metrics. 
For example, it is easy to show that all integral probability metrics -- which includes the bounded Lipschitz 
metric, maximum mean discrepancy, and 1-Wasserstein distance -- are jointly convex (see Lemma S1 in the Supplementary Materials). 
\cref{assump:metric-discr-conditions}(c) is a natural requirement that the divergence estimator is consistent. 
Such estimators are well studied for many common discrepancies.
We discuss consistent estimation of the Kullback--Leibler divergence in \cref{sec:kl-estimation}.
\cref{assump:metric-discr-conditions}(d) will typically hold as long as the map $\phi \mapsto F_{\phi}$ is 
well-behaved.
For example, for the Kullback--Leibler divergence estimators described in \cref{sec:kl-estimation} 
and standard maximum mean discrepancy estimators \citep{Gretton:2012,Krause:2023}, 
when $F_{\phi}$ admits a density $f_{\phi}$, it suffices for the map $\phi \mapsto f_{\phi}(x)$ 
to be continuous for $P_{o}$-almost every $x$. 
\cref{assump:metric-discr-conditions}(e) is not overly restrictive. See \citet{Gibbs:2002} for a extensive overview 
of the relationships between common metrics and divergences. 
\cref{assump:metric-discr-conditions}(f) trivially holds for bounded metrics such as bounded Lipschitz metric and integral probability measures with uniformly bounded test functions.

Our second assumption requires the inference algorithm to be sufficiently regular, 
in the sense that, for each fixed number of components $\numcomps$, the parameter estimate should consistently estimate an asymptotic parameter $\allparam_{\star}^{(K)}$.

\begin{assumption}\label{assump:inference-regularity}
	For each $\numcomps > 0$, there exists $\theta_{\star}^{(K)} \in \Theta^{(K)}$ such that
	the inference algorithm estimate $\hat\theta^{(K)} \to \theta_{\star}^{(K)}$
	in probability as $N \rightarrow \infty$, after possibly reordering of components.
\end{assumption}
To simplify notation, we will write $F^{(K)}_{\star k} = F_{\param^{(K)}_{\star k}}$ and $G^{(K)}_{\star k} = G_{\allparam^{(K)}_{\star k}}$, and similarly for their densities. 
\Cref{assump:inference-regularity} holds for most reasonable algorithms, including expectation--maximization, point estimates based on the posterior distribution,
and variational inference \citep{Balakrishnan:2017,Walker:2001,Wang:2019}.
Note that we assume that consistency holds for parameters in the equivalence class induced by component reordering, although we keep
this equivalence implicit in the discussion that follows. 
\Cref{assump:inference-regularity} implies that the empirical data distribution of the $k$th component, 
$\hat{P}^{(K)}_k = |X_{k}(z^{(K)})|^{-1}\sum_{x \in X_{k}(z^{(K)})} \delta_{x}$,
converges to a limiting distribution
\begin{align}
	\tilde{P}^{(K)}_{k}(\dee x) &= \frac{p_{\star}^{(K)}(k\mid x)P_o(\dee x) }{\int p_{\star}^{(K)}(k\mid y)P_o(\dee y)}, \label{eq:phat-def} 
\end{align}
where $p_{\star}^{(K)}(k \mid x) = \pi^{(K)}_{\star k}f_{\star k}^{(K)}(x)/g^{(K)}_{\star}(x)$ is the 
conditional component probability under the limiting model distribution.

Our third assumption concerns the regularity of the data distribution and model.
\begin{assumption}[$\rho$] \label{assump:model-conditions}
	The data-generating distribution $P_o$ and component model family $\mcF$ satisfy the following conditions:
	\begin{enumerate}[label = (\alph*),ref=(\theassumption-\alph*)]
		\item \emph{Meaningful decomposition of the data-generating distribution:} for positive integer $\numcomps_{o}$, $\pi_{o} \in \Delta_{\numcomps_o}$, and distributions $P_{o1},\dots,P_{o\numcomps_{o}}$ on $\mcX$, it holds that $P_{o} = \sum_{k = 1}^{K_o}\pi_{ok}P_{ok}$.
		\item \emph{Accuracy of the parameter estimates for $\numcomps = \numcomps_{o}$:} for each $k=1, \ldots, \numcomps_o$, it holds that 
		$\discr{\tilde{P}_{k}^{(K_{o})} }{F^{(K_{o})}_{\star k}} < \rho$.
		\item \emph{Poor model fit when $\numcomps$ is too small:} for any $\numcomps < \numcomps_o$, it holds that 
		$d(P_o, G_{\star}^{(\numcomps)}) > \phi(\rho) $.
	\end{enumerate}	
\end{assumption}

\cref{assump:model-conditions}(a) formalizes the decomposition of the data-generating distribution into the subpopulations/types/groups that  we aim to recover.
\cref{assump:model-conditions}(b) requires that, when the number of components is correctly specified,
the divergence between the asymptotic empirical component distribution $\tilde{P}^{(\numcomps_{o})}_{k}$ and the asymptotic component  
$F_{\star k}^{(\numcomps_{o})}$ is small for each component $k=1,\ldots,\numcomps_{o}$.
In \cref{subsec:theory-rho-bounds}, we give conditions under which $\discr{P_{ok}^{(K_{o})} }{F^{(K_{o})}_{\star k}} < \rho_{o}$ 
implies the assumption holds for the Kullback--Leibler divergence and integral probability metrics with bounded test functions. %
\cref{assump:model-conditions}(c) formalizes the intuition that model selection will only be successful if, when the number of components is smaller than the true generating process,
the mixture model is a poor fit to the data. 
The necessary degree of mismatch depends on the match between the true and estimated component distributions, which depends on the choice of $\mcD$ and is measured by $\rho$,
and the relationship between $\mathcal{D}$ and $d$, which is described by $\phi$.

The following result provides a general framework to establish when our method consistently estimates $K_{o}$. %
For notational clarity, let $\hat{\numcomps}_{\numobs}(\rho) = \hat{\numcomps} = \argmin_{\numcomps} \mathcal{R}^{\rho,\lambda}(\hat\allparam^{(\numcomps)}; z^{(\numcomps)}, \data{1:\numobs})$ denote the structurally aware robust model choice. 

\begin{theorem} \label{thm:main}
	If \crefrange{assump:metric-discr-conditions}{assump:model-conditions}($\rho$) hold, then $\Pr\{\hat{\numcomps}_{\numobs}(\rho) = K_{o}\} \to 1$ as $N \to \infty$. 
\end{theorem}

To prove \cref{thm:main} we establish two facts.
First, the (unregularized) structurally aware loss satisfies $\Pr\{\mcR^{\rho}(\allparam^{(\numcomps_{o})}; z^{(\numcomps_{o})}, \data{1:\numobs}) = 0 \} \to 1$ as $\numobs \to \infty$.
Second, for $\numcomps < \numcomps_o$,  $\mcR^{\rho}(\allparam^{(\numcomps)}; z^{(\numcomps)}, \data{1:\numobs}) \to \infty$ in probability
as $\numobs \to \infty$ 
Therefore, for $\numcomps > \numcomps_o$, asymptotically 
$\mcR^{\rho,\lambda}(\allparam^{(\numcomps)}; z^{(\numcomps)}, \data{1:\numobs})  \ge \mcR^{\rho,\lambda}(\allparam^{(\numcomps_{o})}; z^{(\numcomps_{o})}, \data{1:\numobs})  + \lambda (\numcomps - \numcomps_{o})$, so the minimum is asymptotically attained at $\hat{\numcomps}=\numcomps_o$.

\Cref{thm:main} leaves two important questions open.
First, what choices of $\mcD$ and $d$ can satisfy \cref{assump:metric-discr-conditions} (and for what choice of $\phi$)?
While \cref{assump:model-conditions} requires that $\discr{\tilde{P}_{k}^{(K_{o})}}{F^{(K_{o})}_{\star k}} < \rho$,
it is more natural to require a bound on $\discr{P_{ok}}{F^{(K_{o})}_{\star k}}$.
So, the second question is, Does a bound on the latter imply the former?
We address each of these questions in turn. 
\subsection{Choice of Divergence} 
There are many possible choices for the divergence, including the Kullback--Leibler divergence \citep{Joyce:2011}, 
maximum mean discrepancy \citep{Sriperumbudur:2010}, Hellinger distance \citep{Gibbs:2002},
or Wasserstein metric \citep{Villani:2009}. 
We next show that both the Kullback--Leibler divergence and mean maximum discrepancy can satisfy \cref{assump:metric-discr-conditions}(a,b,e) with an appropriate choice of $d$. 
Specifically, we will make use of integral probability metrics.
Given a collection $\mcH$ of real-valued functions on $\mcX$, 
the corresponding integral probability metric is 
\[ \label{eq:IPM}
d_{\mcH}(P,Q) = \sup_{h \in \mcH}\left|\int h(x)P(\dee x) - \int h(y) Q(\dee y)\right|.
\]

If $\mcD$ is the Kullback--Leibler divergences, we choose $d$ to be the bounded Lipschitz metric.
Assume that $\mcX$ is equipped with a metric $m$ and define the bounded Lipschitz norm $\BLnorm{h} = \Vert h \Vert_{\infty} + \Vert h \Vert_{L}$, where $\Vert{h}\Vert_{L} = \sup_{x \ne y}|h(x) - h(y)|/m(x, y)$ and $\Vert{h}\Vert_{\infty} = \sup_{x} |h(x)|$.
Letting $\mcH = \mcH_{\mathrm{BL}} = \{h : \Vert h \Vert_{\mathrm{BL}}\le 1\}$ gives the bounded Lipschitz metric
$\blmetric = d_{\mcH_{\mathrm{BL}}}$.

\begin{proposition} 	\label{coro:KL}
	If $d(P, Q) = \blmetric(P,Q)$ and $\mathcal{D}(P \mid Q) = \kl{P}{Q}$, then 
	\cref{assump:metric-discr-conditions}(a,b,e) holds with $\phi(\rho) = (\rho/2)^{1/2}$.
\end{proposition}

If we choose the divergence to be a mean maximum discrepancy, $d$ can be the same.
Let $\mcK : \mcX \times \mcX \to \mathbb{R}$ denote a positive definite kernel. 
Denote the reproducing kernel Hilbert space with kernel $\mcK$ as $\mcH_{\mcK}$.
Denote its inner product by $\langle{\cdot},{\cdot}\rangle_{\mcK}$ and norm by $\|{\cdot}\|_{\mcK}$. 
Letting $\mcH = \mcB_{\mcK} = \{ h \in \mcH_{\mcK} :  \|h\|_{\mcK} \le 1\}$, the unit ball, gives the
mean maximum discrepancy $d_{\mathrm{MMD}} = d_{\mcB_{\mcK}}$.
\begin{proposition} 	\label{coro:MMD}
	If $\mcK$ is chosen such that $d_{\mathrm{MMD}}$ metrizes weak convergence, and 
	$d(P, Q) = \mathcal{D}(P \mid Q) = d_{\mathrm{MMD}}(P,Q)$, 
	then \cref{assump:metric-discr-conditions}(a,b,e) holds with $\phi(\rho) = \rho$.
\end{proposition}

For conditions under which $d_{\mathrm{MMD}}$ metrizes weak convergence, see 
\citet{Sriperumbudur:2010,Simon:2020}.

\subsection{Bounding the Component-level Discrepancy}
\label{subsec:theory-rho-bounds}

\Cref{assump:model-conditions}($\rho$)(b), requires the discrepancy between the limiting empirical component distribution and the model component 
be less than $\rho$. %
However, such a requirement is not completely satisfactory since $\tilde{P}_{k}$ depends on both $P_{o}$ and the mixture model family. 
A more intuitive and natural assumption would bound the divergence between true component distribution and model component -- that is, be of the 
form $\mathcal{D}(P_{ok} \mid F_{\star k}^{(\numcomps_{o})}) < \rho_o$ for some $\rho_{o} > 0$. 
In this section, we show how to relate $\rho_o$ to $\rho$ for the Kullback--Leibler divergence and integral probability metrics, including the maximum mean discrepancy.

Define the conditional component probabilities for the true generating distribution as
\begin{equation}
p_{o}(k \mid x) = \frac{\pi_{ok}p_{ok}(x)}{p_{o}(x)}.
\end{equation}

We rely on the following assumption.
\begin{assumption} \label{assump:bounded-ratios}
	There exist constants $\epsilon_{\pi}, \epsilon_{z}>0$ such that for all 
	$k=1,\ldots, \numcomps_{o}$,
	
\begin{equation}
	\frac{\pi_{\star k} }{\pi_{ok} }  < 1+\epsilon_{\pi}, \qquad
\frac{p_{\star}^{(\numcomps_{o})}(k \mid x)}{p_{o}(k \mid x)}  < 1+\epsilon_{z}.
\end{equation}

\end{assumption}
We first consider the Kullback--Leibler divergence case.
\begin{proposition} \label{prop:kl-upper-bound}
	If $\mathcal{D}(P \mid Q) = \kl{P}{Q}$, \cref{assump:bounded-ratios} holds, 
	and $\kl{P_{ok}}{F_{\star k}^{(\numcomps_{o})}} < \rho_{o}$ for $k=1,\ldots, \numcomps_{o}$, 
	then \Cref{assump:model-conditions}($\rho$)(b) holds for $\rho = (1+\epsilon_{\pi})(1+\epsilon_{z})[\rho_{o} + \log(1+\epsilon_{\pi})(1+\epsilon_{z})]$. 
\end{proposition}

For integral probability measures, we will focus on the common case where the test functions are uniformly bounded. 
\begin{proposition} \label{prop:IPM-metric-upper-bound}
	If $\mathcal{D}(P \mid Q) = d_{\mcH}(P, Q)$, \cref{assump:bounded-ratios} holds, there exists $M > 0$ such that $\|h\|_{\infty} < M$ for all $h \in \mcH$, 
	and $d_{\mathcal{H}}(P_{ok},F^{(\numcomps_{o})}_{\opt k}) < \rho_{o}$ for $k=1,\ldots, \numcomps_{o}$, then
	\Cref{assump:model-conditions}($\rho$)(b) holds for $\rho = M\left(\epsilon_{\pi}+\epsilon_{z}+\epsilon_{\pi}\epsilon_{z}\right) +\rho_{o}$.
\end{proposition}

\cref{prop:IPM-metric-upper-bound} provides justification for a broad class of metrics including the bounded Lipschitz metric, the total variation distance,
and the Wasserstein distance when $\mcX$ is a compact metric space. 
In addition, it applies to maximum mean discrepancies as long as $\mcK(x, x)$ is uniformly bounded since,
for $h \in \mcB_{\mcK}$, 
\begin{equation}
	h(x) = \langle h, \mcK(x, \cdot) \rangle_{\mcK} \le \Vert h \Vert_{\mcK}  \Vert \mcK(x,\cdot )\Vert_{\mcK}  = \mcK(x,x)^{1/2}, %
\end{equation} 
so we may take $M = \sup_{x} \mcK(x,x)^{1/2}$.

\section{Practical Considerations}
\label{sec:guidance}

\subsection{Computation}
\label{sec:algorithm}

Our model selection framework can be used with any parameter estimation algorithm (e.g., expectation--maximization, Markov chain Monte Carlo, or variational inference).
Thus, the choice of algorithm can be based the statistical and computational considerations relevant to specific problem at hand. 
In our experiments we use an expectation--maximization algorithm to approximate the maximum likelihood parameter estimates for each 
possible choice of $\numcomps$.

Once parameter estimates have been obtained, the other major cost is computing the penalized structurally aware loss 
$\mathcal{R}^{\rho,\lambda}(\hat\allparam^{(\numcomps)}; z^{(\numcomps)}, \data{1:\numobs})$, which is generally dominated by 
the computation of the divergence estimator. %
We discuss the particular case of Kullback--Leibler divergence estimation in \cref{sec:kl-estimation}.
Numerous approaches exists for estimating other metrics such as the maximum mean discrepancy \citep{Gretton:2012,Krause:2023}.

\subsection{Choosing  \texorpdfstring{$\rho$}{rho}} \label{sec:choosing-rho}

While the results from \cref{sec:theory} show that, with an appropriate choice of $\rho$, our method consistently estimates the number of mixture model components, 
those results do not suggest a way to select $\rho$ in practice. 
Hence, we propose two complementary approaches to selecting $\rho$ that 
take advantage of the fact that the structurally aware loss is a piecewise linear function of $\rho$. 
Therefore, given a fitted model for each candidate $\numcomps$, we can easily compute the structurally aware loss for all values of $\rho$. 

Our first approach aims to leverage domain knowledge. %
Specifically, it is frequently the case that some related datasets are available with ``ground-truth'' labels either 
through manual labeling or via \emph{in silico} mixing of data where group labels are directly observed \citep[see, e.g.,][]{Souto:2008}.
In such cases an appropriate $\rho$ value or range of candidate values for one or more such datasets with ground-truth labels 
can be determined by maximizing a clustering accuracy metric such as $F$-measure. 
Because $\rho$ quantifies the divergence between the true component distributions and the model estimates, 
we expect the values found using this approach will generalize to new datasets that are reasonably similar. 
This approach is employed in our real-data experiments with flow cytometry data, discussed in Section \ref{sec:experiments}. 

For applications where there are no related datasets with ground-truth labels available, we 
propose a second, more heuristic approach.  %
After estimating the model parameters for each fixed $\numcomps=1,\ldots, \numcomps_{\max}$
and computing all component-wise divergences, we plot the penalized loss as a function of $\rho$ for each $\numcomps \in \{1,\dots,\numcomps_{\max}\}$.
The optimal model is determined by identifying the number of components which is best over a wide range of $\rho$ values,
with $\rho$ as small as possible. 
The idea behind this selection rule is to identify the first instance of stability, indicating that increasing $\rho$ further doesn't notably improve the loss. 
However, subsequent stable regions that appear afterward might introduce too much tolerance, potentially resulting in model underfitting.
This approach is similar in spirit to the one introduced for heuristically selecting the $\alpha$ parameter for the coarsened posterior \citep{Miller:2019}.

We illustrate our second approach using a Poisson mixture model simulation study.
Suppose data $x_1, \ldots, x_{\numobs}$ is generated from a negative binomial mixture $P_o = \sum_{k=1}^{\numcomps_o}, \pi_{ok}\distNBinom(m_k, p_k)$. %
The assumed model is $G_{\theta} = \sum_{k=1}^{\numcomps} \pi_k \distPoiss(\param_k)$.
We set $\numobs=20\,000, \numcomps_o = 3, m = (55, 75, 100), p = (0.5, 0.3, 0.5)$, and  $\pi_o = (0.3, 0.3, 0.4)$.
We use the plug-in Kullback--Leibler estimator (see \cref{eq:plug-in-KL} in \cref{sec:kl-estimation} below).
Based on \cref{fig:poismm}(left), the first wide and stable region corresponds to the true number of components $K = 3$. 
The observed data and fitted model distribution for $K = 3$ are shown in \cref{fig:poismm}(right).

\begin{figure}[tp]
	\centering	
	\subfloat{%
		\includegraphics[width=65mm,height=40mm]{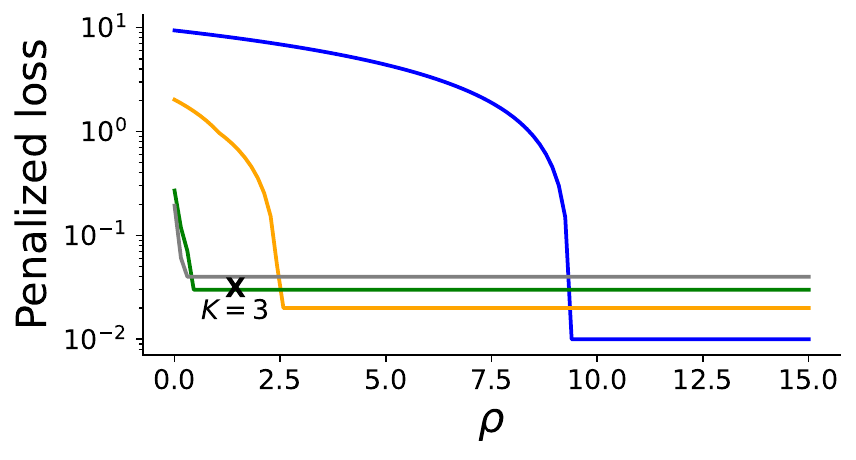}}
	\subfloat{%
		\includegraphics[width=65mm,height=40mm]{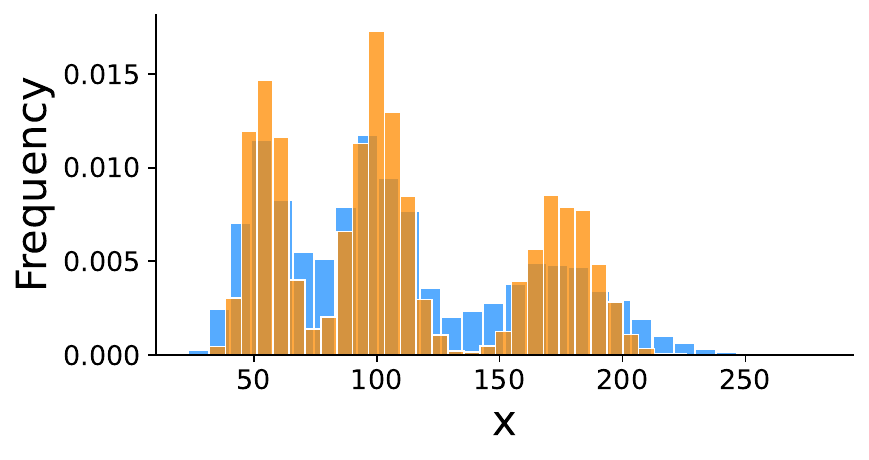}}
	\caption{Fitting a Poisson mixture model to data from a mixture of negative binomial distributions (\cref{sec:choosing-rho}). \textbf{Left:} Penalized loss plot for $\numcomps=1$ (blue), $\numcomps=2$ (orange), $\numcomps=3$ (green) and $\numcomps=4$ (gray). The cross mark indicates the first wide stable region and is labeled with the number of clusters our method selects. \textbf{Right:} Estimated model distribution (orange) compared to the observed data (blue).} 
	\label{fig:poismm}
\end{figure}

\subsection{Choosing \texorpdfstring{$\lambda$}{lambda}}

The penalized structurally aware loss function requires selecting the penalty parameter $\lambda$. 
In practice the primary purpose of the penalty is to select the smallest value of $\numcomps$ that results in $\mcR^{\rho}(\allparam^{(\numcomps)};z_{1:\numobs}, \data{1:\numobs}) = 0$.
So, we recommend setting $\lambda$ to a value that is small relative to the estimates of the divergence.
For example, we set $\lambda=0.01$ in our simulation studies and set $\lambda=10$ in the real-data experiments 
so as to ensure the $\rho$ versus loss plots were easy to read. In particular, choosing smaller values would not have changed the results.

\subsection{Estimation of Kullback--Leibler divergence}
\label{sec:kl-estimation}

Our theoretical results naturally require a consistent estimator of the divergence. %
Thus, we briefly cover the question of how best to estimate the Kullback--Leibler divergence in practice. 
For further details, refer to the Supplementary Materials.
We consider a general setup with observations $y_{1:\numobs} = (y_{1},\dots, y_{\numobs}) \in \mcX^{\otimes \numobs}$ independent, identically distributed 
from a distribution $P$.
First consider the case where $\mcX$ is countable, and let $\numobs(x) = \#\{ n \in \{1,\dots, \numobs\} \mid y_{n} = x\}$ denote the number of observations 
taking the value $x \in \mcX$. 
Letting $Q$ be a distribution with probability mass function $q$, we can use the plug-in estimator for $\kl{P}{Q}$,

\begin{equation}
	\klest{y_{1:\numobs}}{F} 
	= \sum_{x \in \mcX} \frac{\numobs(x)}{\numobs} \log \left\{\frac{\numobs(x)}{\numobs q(x)}\right\},
	\label{eq:plug-in-KL}
\end{equation}

which is consistent under modest regularity conditions \citep{Paninski:2003}.

Next we consider the case of general distributions on $\mcX = \mathbb{R}^{D}$,
when estimation of Kullback--Leibler divergence is less straightforward. 
One common approach is to utilize $k$-nearest-neighbor density estimation.
For $r > 0$, let $V_{D}(r) = \frac{\pi^{D/2}}{\Gamma(D/2+1)}r^{D}$ denote the volume of an 
$D$-dimensional ball of radius $r$ and let $r_{k,n}$ denote the distance to the $k$th nearest neighbor of $y_n$.
Following the same approach as \citet{Zhao:2020} and assuming the distribution $Q$ has Lebesgue density $q$,
we obtain a one-sample estimator for $\kl{P}{Q}$:
\begin{equation}
		\klestsub{b}{k}{y_{1:\numobs}}{Q}
		= \frac{1}{N}\sum_{n=1}^{\numobs} \log\left\{\frac{k/(\numobs-1)}{V_{D}(r_{k,n}) q(y_n)}\right\}.
		\label{eq:stare-knn-kl-est}
\end{equation}
As we discuss in the Supplementary Materials, for fixed $k$, the estimator in \cref{eq:stare-knn-kl-est} is asymptotically biased. 
However, it is easy to correct this bias, leading to the unbiased, consistent estimator 
\begin{equation}
	\klestsub{u}{k}{y_{1:\numobs}}{Q} = \klestsub{b}{k}{y_{1:\numobs}}{Q} - \log k + \psi(k),
\end{equation}
where $\psi(k)$ denotes the digamma function.
Another way to construct a consistent estimator is to let $k = k_{\numobs}$ depend on the data size $\numobs$,
with $k_{\numobs}  \to \infty$ as $\numobs \to \infty$. 
A canonical choice is $\klestsub{b}{k_{\numobs}}{y_{1:\numobs}}{Q}$ with $k_{\numobs}=\numobs^{1/2}$. 

We compare the three estimators for various dimensions in Supplementary Materials. 
Our results show that the bias-corrected estimator slightly improves the biased version when $\numobs \ge 5000$, 
while the adaptive estimator with $k_\numobs = \numobs^{1/2}$ has the most reliable performance  when $D=4$.
Since the data dimensions are relatively low in all our experiments, we use the adaptive estimator with $k_N=N^{1/2}$. %

It is important to highlight that Kullback--Leibler divergence estimators require density estimation, which in general requires the sample size to grow exponentially with the dimension \citep{Donoho:2000}.
This limits the use of such estimators with generic high dimensional data.
However, a general strategy to address this would be to take advantage of some known or inferable 
structure in the distribution to reduce the effective dimension of the problem.
We provide a more detailed illustration of this strategy in \cref{sec:high-dim-simulation} with simulated data that exhibits weak correlations across coordinates.

\section{Numerical Experiments}
\label{sec:experiments}
\subsection{Simulation Study: Skew-normal mixture}
\label{sec:simulation-gauss}

We now provide further details about the motivating example in \cref{sec:motivation}, and illustrate 
that our method selects the correct number of components under a variety of conditions on the level of misspecification
and the relative sizes of the mixture components.
We generate data $x_1, \ldots, x_n \in \mathbb{R}$ from skew-normal mixtures of the form $P_{o} = \sum_{k=1}^{\numcomps_{o}}\pi_{ok}\distSNorm(\mu_{ok}, \sigma_{ok},\gamma_{ok})$, where $\distSNorm(\cdot)$ denotes the skew-normal distribution and $\gamma_{ok}$ denotes the skewness parameter.
The density of the skew-normal distribution $\distSNorm(\mu, \sigma,\gamma)$ is $f(x) = 2\phi(x;\mu,\sigma)\Phi(\gamma x;\mu,\sigma)$, where $\phi(x;\mu,\sigma)$ and $\Phi(x;\mu,\sigma)$ denote the probability density function and the cumulative distribution function of $\distNorm(\mu,\sigma)$ respectively.
We model the data using  Gaussian mixture model $G_{\allparam} = \sum_{k=1}^{\numcomps}\pi_{k}\distNorm(\mu_{k}, \sigma_{k})$.
The bigger $|\gamma_{ok}|$, the larger the deviation from the Gaussian distribution,
hence introducing a higher degree of misspecification.

In \cref{sec:motivation}, we considered the case of two clusters of similar sizes. We set $\pi_{o} = (0.5, 0.5)$, $\mu_{o} = (-3,3)$, and $\sigma_{o} = (1,1)$ for the two scenarios in \cref{fig:motivate-comparison}: $\gamma_o = (-10,-1)$ (denoted \texttt{different}) and $\gamma_o = (-10,-10)$ (denoted \texttt{same}). 
We now compare our approach to the coarsened posterior with data from 
two-component mixtures of different cluster sizes. %
We set $\pi_{o} = (0.95, 0.05)$, $\mu_{o} = (-3,3)$, and $\sigma_{o} = (1,1)$ for the following three scenarios:
$\gamma_o = (-10,-1)$ (denoted \texttt{large-small}),  $\gamma_o = (-1,-10)$ (denoted \texttt{small-large}),
and  $\gamma_o = (-10,-10)$ (denoted \texttt{large-large}).

\begin{figure}[tp]
	\centering	
	\subfloat{\label{fig:gauss-cpos-pdfs-1}\includegraphics[width=50mm]{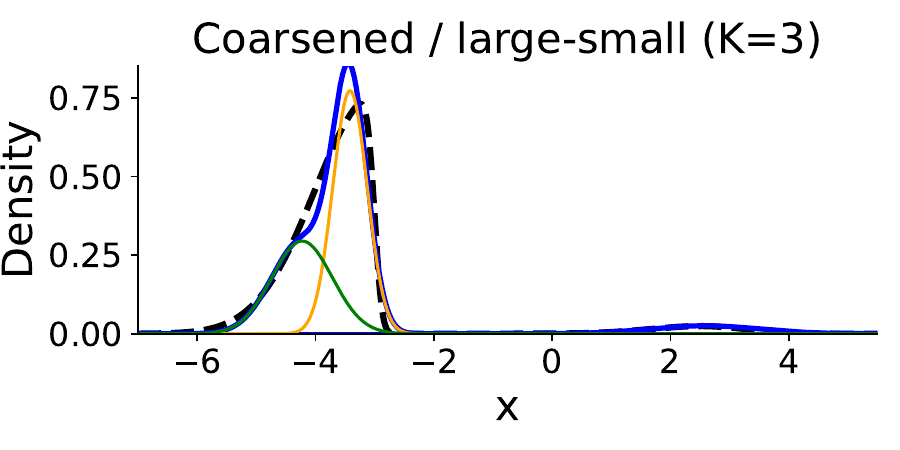}}
	\subfloat{\label{fig:gauss-cpos-pdfs-2}\includegraphics[width=50mm]{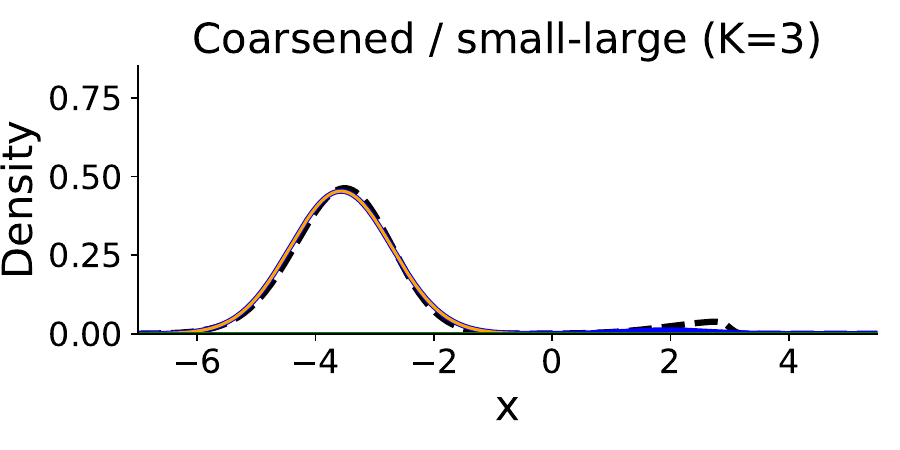}}
	\subfloat{\label{fig:gauss-cpos-pdfs-3}\includegraphics[width=50mm]{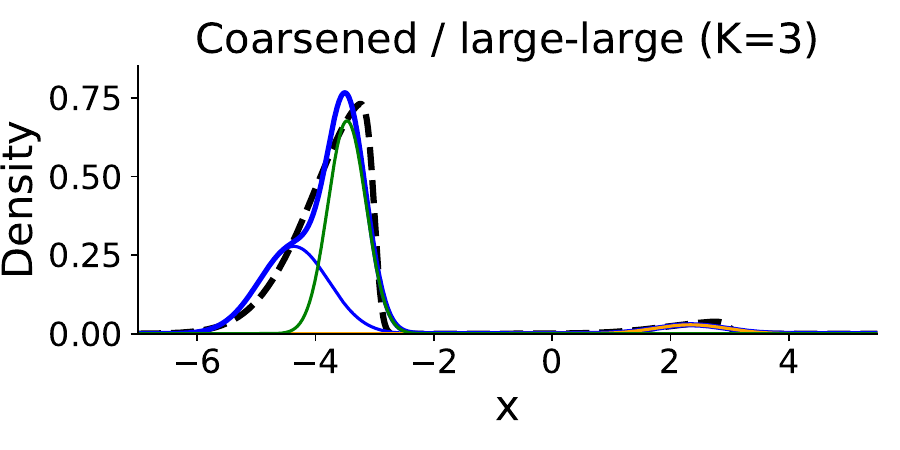}}
	\\
	\subfloat{\label{fig:gauss-stare-pdfs-1}\includegraphics[width=50mm]{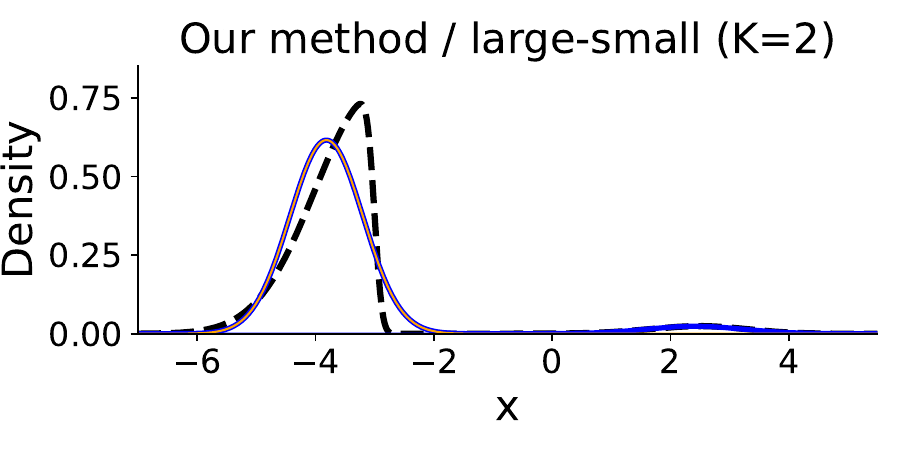}}
	\subfloat{\label{fig:gauss-stare-pdfs-2}\includegraphics[width=50mm]{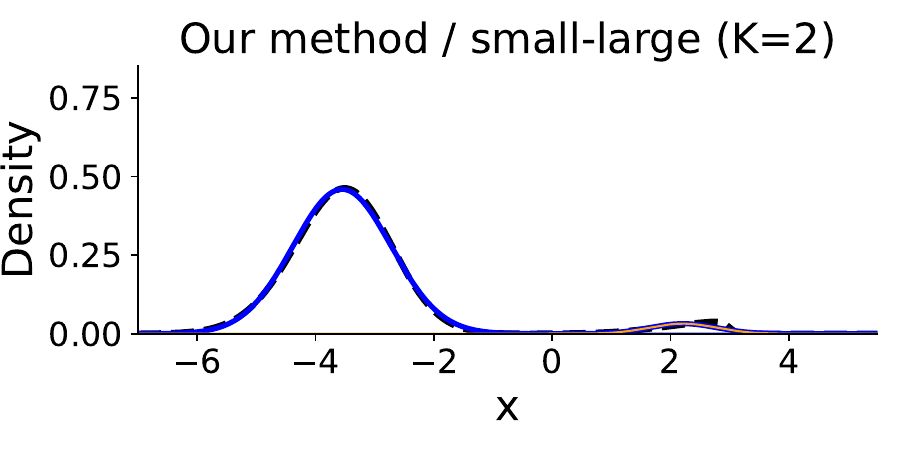}}
	\subfloat{\label{fig:gauss-stare-pdfs-3}\includegraphics[width=50mm]{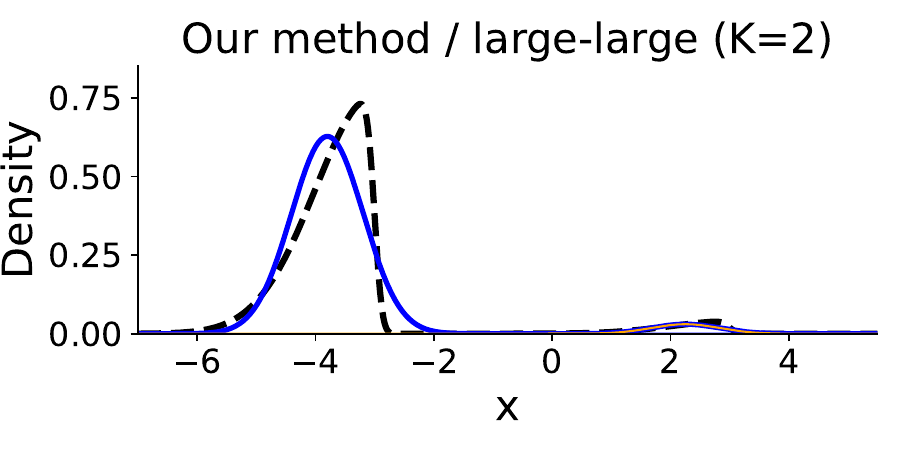}}
	\\
	\subfloat{\label{fig:gauss-stare-penloss-1}\includegraphics[width=50mm]{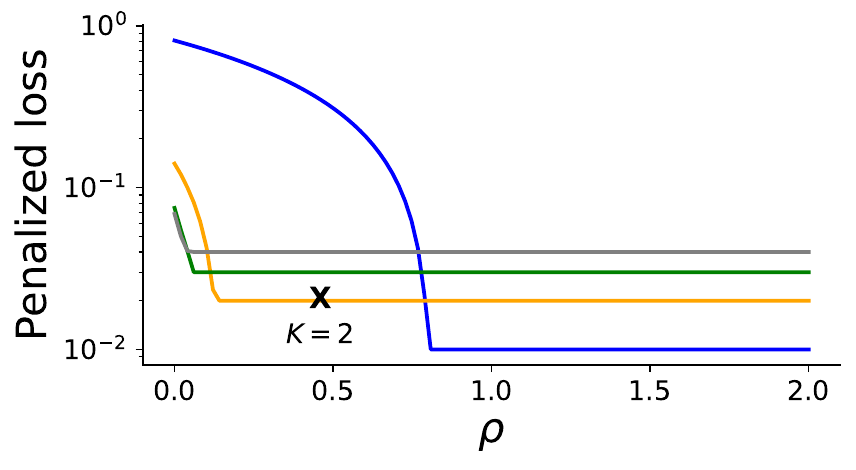}}
	\subfloat{\label{fig:gauss-stare-penloss-2}\includegraphics[width=50mm]{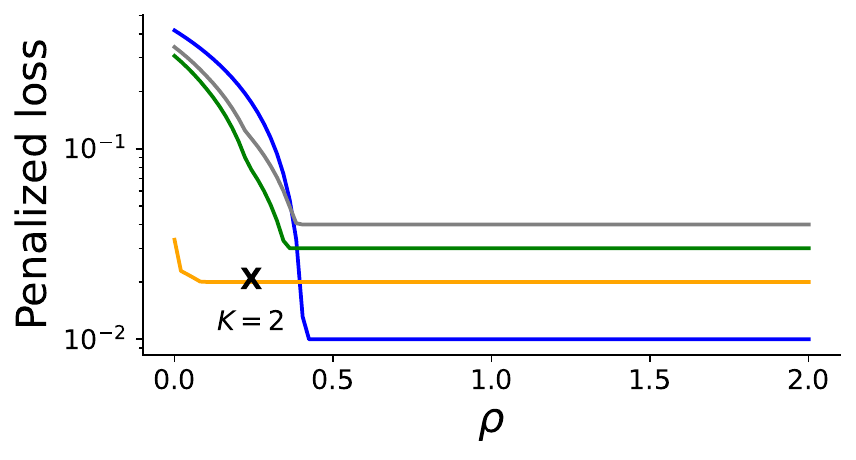}}
	\subfloat{\label{fig:gauss-stare-penloss-3}\includegraphics[width=50mm]{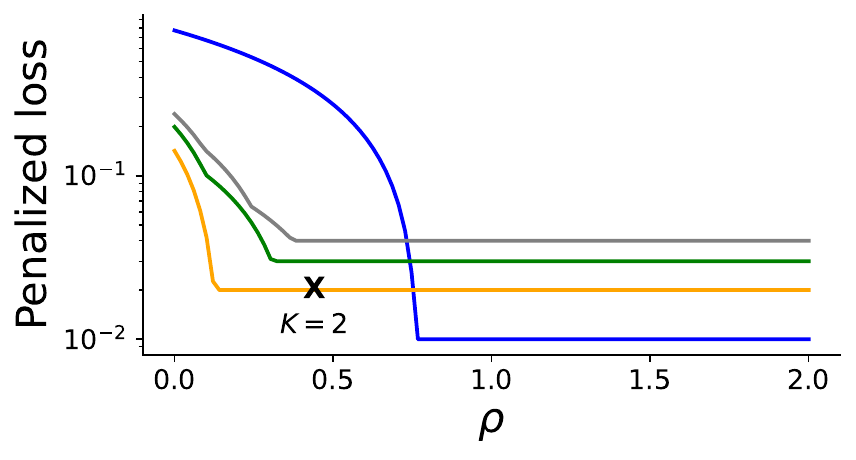}}

	\caption{
		Comparison between the coarsened posterior and our method when using a Gaussian mixture model to fit data generated from a mixture of 
		skew-normal distributions. 
		\textbf{First row:} Densities of the model and components selected using the coarsened posterior (solid lines) and the density of the data distribution (dashed line). The title specifies the data-generating distribution and the number of components selected. 
		In the middle plot of the first row, the minor cluster contains two components.
		\textbf{Second row:} Densities of the model and components selected using our structurally aware robust method. 
		\textbf{Third row:} Penalized loss plots, where the cross mark indicates the first wide stable region and is labeled with the number of clusters our method selects. 
		Line colors correspond to different $\numcomps$ values. See caption in \cref{fig:poismm} for details.} 
	\label{fig:gauss-simulation}
\end{figure}

For the coarsened posterior, we calibrate the hyperparameter $\alpha$ following the procedure from \citet[Section 4]{Miller:2019}.
First, we use Markov chain Monte Carlo to approximate the coarsened posterior for $\alpha$ values ranging from $10$ to $10^5$.
Then, we select the coarsened posterior with the $\alpha$ value at the clear cusp that indicates a good fit and low complexity. 
See the Supplementary Materials for further details and calibration plots. 
As shown for \texttt{large-small} and \texttt{large-large} in \cref{fig:gauss-simulation}, when the 
larger cluster has large misspecification, the coarsened posterior introduces one additional cluster to explain the larger cluster.
For the \texttt{small-large} case, when the larger cluster exhibits a small degree of misspecification, 
the coarsened posterior introduces one additional cluster to explain the smaller cluster.

Our method correctly calibrates the model mismatch cutoff $\rho$ using the penalized loss plots shown in \cref{fig:gauss-simulation},
as in all cases $\numcomps = 2$ corresponds to the first wide, stable region.
By the density plots in the middle column, we can see that our structurally aware robust model selection method 
is able to properly trade off a worse density estimate for better model selection.
Our approach also enjoys improved computational efficiency compared to coarsened posterior (see discussion in \cref{sec:motivation}). 
This is verified in the runtime comparison that running one scenario using our structurally aware robust model selection method (code in Python) takes about 1 minute 
while using the coarsened posterior (code in Julia) takes 140 minutes, despite Julia generally being much faster than Python in scientific computing applications \citep{Perkel:2019}.

\subsection{High Dimensional Study}
\label{sec:high-dim-simulation}

As discussed in Section \ref{sec:kl-estimation}, the Kullback--Leibler $k$-nearest-neighbor estimator becomes less accurate with increasing data dimension. 
While a general solution is unlikely to exist, we illustrate one approach to address this challenge.
Specifically, if we believe the coordinates are likely to be only weakly correlated, we can employ the $k$-nearest-neighbor method on each coordinate by assuming the coordinates are independent.
\begin{figure}[tp]
	\centering	
	\subfloat{\label{fig:high-dim-stare-loss}\includegraphics[width=65mm]{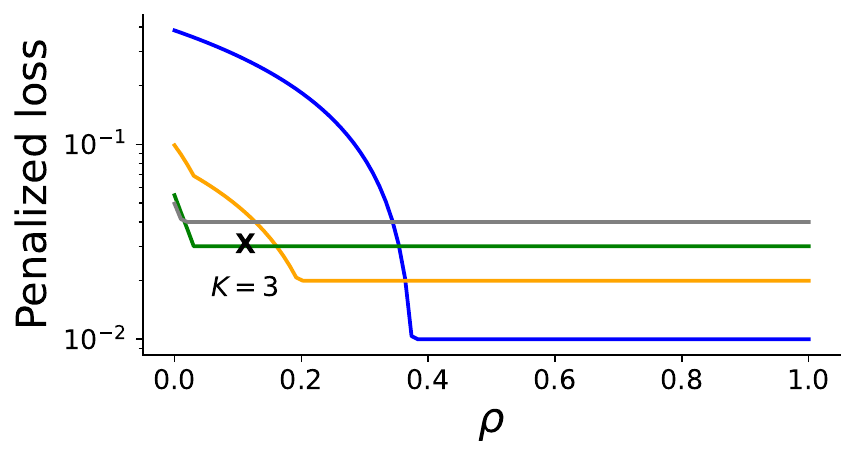}}
	\subfloat{\label{fig:high-dim-true}\includegraphics[width=65mm]{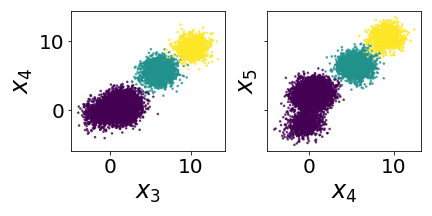}}
	
	\caption{Application of our method to simulated high dimensional data. \textbf{Left:} The penalized loss plot 
		for determining the number of componetnts. See \cref{fig:poismm} for description of line colors. 
		\textbf{Right:} Selected two-dimensional projections of ground truth.} 
	\label{fig:high-dim}
\end{figure}

We generalize the simulation example from \cref{sec:simulation-gauss} to generate data from multivariate skew normal mixtures
and fit a multivariate normal mixture model. 
Suppose $x_1, \ldots,x_{\numobs} \in \mathbb{R}^{D}$ is generated from $P_{o} = \sum_{k=1}^{\numcomps_{o}}\pi_{ok}\distSNorm(m_{ok}, \Sigma_{ok}, \gamma_{ok})$, where $\gamma_{ok}$ is a $D$-dimensional vector controlling the skewness of the distribution, $m_{ok}$ and $\Sigma_{ok}$ are the means and covariance matrix of the $k$th component. 
The density for multivariate skew normal distribution is $f(x; m, \Sigma, \gamma) = 2\phi(x; m, \Sigma)\Phi(\gamma \odot x; m, \Sigma)$, where $\phi(x; m, \Sigma)$ and $\Phi(x; m, \Sigma)$ are the probability density function and cumulative distribution function of multivariate normal $\distNorm(m, \Sigma)$ respectively.
For each covariance matrix, we introduce weak correlations by letting $\Sigma_{ok} = \Sigma$,
where $\Sigma_{ij}=\exp\{-(i-j)^2/\sigma^2\}$ and $\sigma$ controls the strength of the correlation. 
We set $\sigma=0.6$ and sample  $N=10\,000$ observations of dimension $D=50$.
As shown in \cref{fig:high-dim}(left), the wide and stable region corresponds to the true number of components $K = 3$.
\cref{fig:high-dim}(right) illustrates the value of model-based clustering, particularly in high dimensions: 2-dimensional 
projections of the data give the appearance of there being four clusters in total, when in fact there are only three.

\subsection{Application to Flow Cytometry Data}
\label{sec:flow-cytometry}

Flow cytometry is a technique used to analyze properties of individual cells in a biological material.
Typically, flow cytometry data measures 3--20 properties of tens of thousands cells. 
Cells from distinct populations tend to fall into clusters and discovering cell populations by identifying clusters 
is of primary interest in practice.
Usually, scientists identify clusters manually which is labor-intensive and subjective. 
Therefore, automated and reliable clustering inference is invaluable. 
We follow the approach of \citet{Miller:2019} and test our method on the same 12 flow cytometry datasets originally from a longitudinal study of graft-versus-host disease %
in patients undergoing blood or marrow transplantation \citep{Brinkman:2007}. 
For these datasets, we calibrate $\rho$ using the first 6 datasets. %
We take the manual clustering as the ground truth.
The datasets consist of $D=4$ dimensions and varying number of observations for each dataset.
As shown in \cref{fig:flowcyt-train}, all training datasets 1-6 have a nearly identical trend of
clustering accuracy as a function of $\rho$. 
The averaged F-measure achieves the maximum when $\rho \approx 1.16$, which is a point of maximum F-measure for all 6 datasets. 
The consistency of our method compares favorably to using coarsening, where drastically different $\alpha$ values maximize the F-measure \citep[Figure 5]{Miller:2019}.
These results provide evidence that our approach is taking better advantage of the common structure and degree of misspecification across datasets.

\begin{figure}[tp]
	\centering	
	\includegraphics[width=80mm]{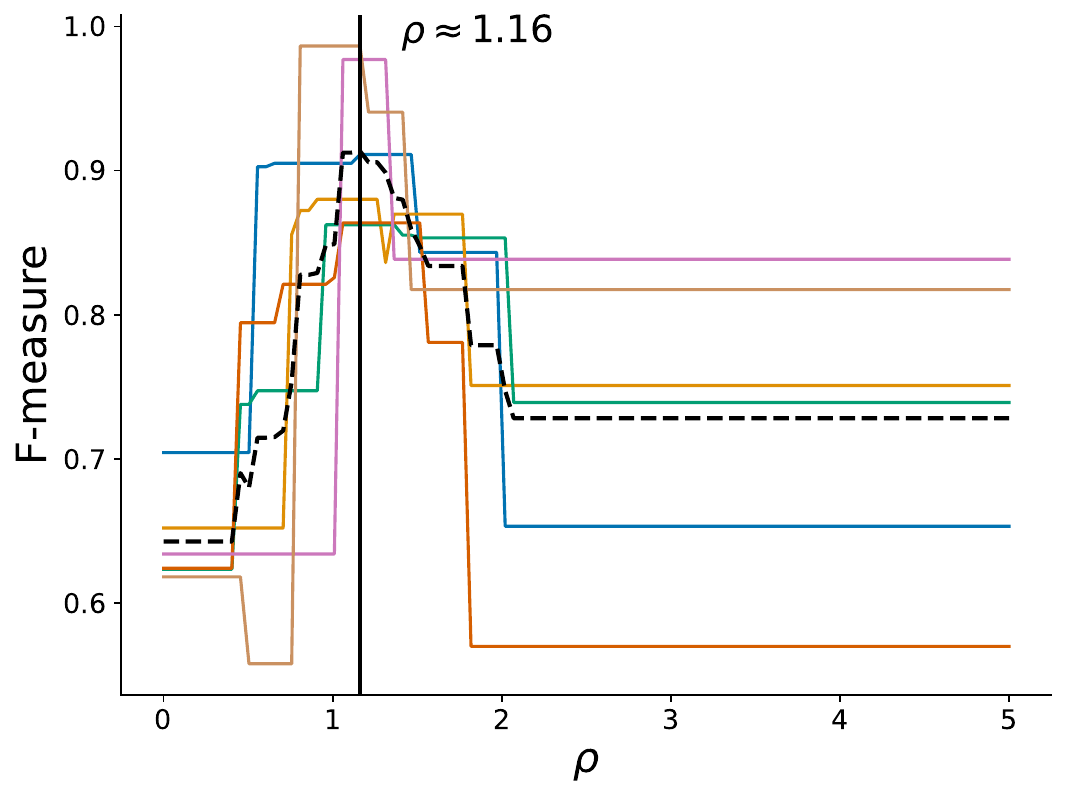}
	\caption{$\rho$ versus F-measure for training datasets 1--6 (solid lines). The black dashed line indicates averaged F-measure over the training datasets.} 
	\label{fig:flowcyt-train}
\end{figure}

For test datasets 7-12, we propose to pick the value of $\numcomps$ that achieves a stable structurally aware loss and uses a value of $\rho$ that is as close
as possible to the estimated $\rho$ values of $1.16$. 
As shown in \cref{tab:flowcyto-test}, our method provides essentially the same average accuracy as the coarsened posterior
while being substantially more computationally efficient, despite using a much slower programming language (2 hours using Python versus 30 hours using Julia). %

\begin{table}[h]
	\centering
	\caption{F-measures on flow cytometry test datasets 7-12}
	\renewcommand{\arraystretch}{1.2}
	\begin{tabular}{ccccccc}
		\hline
		& 7             & 8             & 9             & 10            & 11            & 12            \\ \hline
		Structurally aware                                                        & 0.63          & \textbf{0.92} & \textbf{0.94} & \textbf{0.99} & \textbf{0.99} & \textbf{0.98} \\ \hline
		\begin{tabular}[c]{@{}c@{}}Coarsened \end{tabular} & \textbf{0.67} & 0.88          & \textbf{0.93} & \textbf{0.99} & \textbf{0.99} & \textbf{0.99} \\ \hline
	\end{tabular}
	\label{tab:flowcyto-test}
\end{table}

\section*{Acknowledgement}

J.~Li and J.~H.~Huggins were partially supported by the National Institute of General Medical Sciences of the National Institutes of Health
as part of the Joint NSF/NIGMS Mathematical Biology Program.
The content is solely the responsibility of the authors and does not necessarily represent the official views of the National Institutes of Health.

\bibliographystyle{imsart-nameyear}
\bibliography{paper-ref}

\begin{thebibliography}{38}
% BibTex style file: imsart-nameyear.bst, 2017-11-03
% Default style options (sort=1,type=nameyear).
% Used options (sort=1,type=nameyear).

\bibitem[\protect\citeauthoryear{{Ghorbani Afkhami}, Azarnia and
  Tinati}{2016}]{Ghorbani:2016}
\begin{barticle}[author]
\bauthor{\bsnm{{Ghorbani Afkhami}},~\bfnm{Rashid}\binits{R.}},
  \bauthor{\bsnm{Azarnia},~\bfnm{Ghanbar}\binits{G.}} \AND
  \bauthor{\bsnm{Tinati},~\bfnm{Mohammad~Ali}\binits{M.~A.}}
(\byear{2016}).
\btitle{Cardiac arrhythmia classification using statistical and mixture
  modeling features of ECG signals}.
\bjournal{Pattern Recognition Letters}
\bvolume{70}
\bpages{45-51}.
\bdoi{https://doi.org/10.1016/j.patrec.2015.11.018}
\end{barticle}
\endbibitem

\bibitem[\protect\citeauthoryear{Archambeau and
  Verleysen}{2007}]{Archambeau:2007}
\begin{barticle}[author]
\bauthor{\bsnm{Archambeau},~\bfnm{Cédric}\binits{C.}} \AND
  \bauthor{\bsnm{Verleysen},~\bfnm{Michel}\binits{M.}}
(\byear{2007}).
\btitle{Robust Bayesian clustering}.
\bjournal{Neural Networks}
\bvolume{20}
\bpages{129-138}.
\bdoi{https://doi.org/10.1016/j.neunet.2006.06.009}
\end{barticle}
\endbibitem

\bibitem[\protect\citeauthoryear{Balakrishnan, Wainwright and
  Yu}{2017}]{Balakrishnan:2017}
\begin{barticle}[author]
\bauthor{\bsnm{Balakrishnan},~\bfnm{Sivaraman}\binits{S.}},
  \bauthor{\bsnm{Wainwright},~\bfnm{Martin~J.}\binits{M.~J.}} \AND
  \bauthor{\bsnm{Yu},~\bfnm{Bin}\binits{B.}}
(\byear{2017}).
\btitle{{Statistical guarantees for the EM algorithm: From population to
  sample-based analysis}}.
\bjournal{The Annals of Statistics}
\bvolume{45}
\bpages{77 -- 120}.
\bdoi{10.1214/16-AOS1435}
\end{barticle}
\endbibitem

\bibitem[\protect\citeauthoryear{Bauer}{2007}]{Bauer:2007}
\begin{barticle}[author]
\bauthor{\bsnm{Bauer},~\bfnm{Daniel~J.}\binits{D.~J.}}
(\byear{2007}).
\btitle{Observations on the Use of Growth Mixture Models in Psychological
  Research}.
\bjournal{Multivariate Behavioral Research}
\bvolume{42}
\bpages{757-786}.
\bdoi{10.1080/00273170701710338}
\end{barticle}
\endbibitem

\bibitem[\protect\citeauthoryear{Bharti et~al.}{2023}]{Krause:2023}
\begin{binproceedings}[author]
\bauthor{\bsnm{Bharti},~\bfnm{Ayush}\binits{A.}},
  \bauthor{\bsnm{Naslidnyk},~\bfnm{Masha}\binits{M.}},
  \bauthor{\bsnm{Key},~\bfnm{Oscar}\binits{O.}},
  \bauthor{\bsnm{Kaski},~\bfnm{Samuel}\binits{S.}} \AND
  \bauthor{\bsnm{Briol},~\bfnm{Francois-Xavier}\binits{F.-X.}}
(\byear{2023}).
\btitle{Optimally-weighted Estimators of the Maximum Mean Discrepancy for
  Likelihood-Free Inference}.
In \bbooktitle{Proceedings of the 40th International Conference on Machine
  Learning}
(\beditor{\bfnm{Andreas}\binits{A.}~\bsnm{Krause}},
  \beditor{\bfnm{Emma}\binits{E.}~\bsnm{Brunskill}},
  \beditor{\bfnm{Kyunghyun}\binits{K.}~\bsnm{Cho}},
  \beditor{\bfnm{Barbara}\binits{B.}~\bsnm{Engelhardt}},
  \beditor{\bfnm{Sivan}\binits{S.}~\bsnm{Sabato}} \AND
  \beditor{\bfnm{Jonathan}\binits{J.}~\bsnm{Scarlett}}, eds.).
\bseries{Proceedings of Machine Learning Research}
\bvolume{202}
\bpages{2289--2312}.
\bpublisher{PMLR}.
\end{binproceedings}
\endbibitem

\bibitem[\protect\citeauthoryear{Bishop and
  Svens{\'e}n}{2004}]{Christopher:2004}
\begin{binproceedings}[author]
\bauthor{\bsnm{Bishop},~\bfnm{Christopher~M}\binits{C.~M.}} \AND
  \bauthor{\bsnm{Svens{\'e}n},~\bfnm{Markus}\binits{M.}}
(\byear{2004}).
\btitle{Robust Bayesian Mixture Modelling.}
In \bbooktitle{ESANN}
\bpages{69--74}.
\bpublisher{Citeseer}.
\end{binproceedings}
\endbibitem

\bibitem[\protect\citeauthoryear{Brinkman et~al.}{2007}]{Brinkman:2007}
\begin{barticle}[author]
\bauthor{\bsnm{Brinkman},~\bfnm{William}\binits{W.}},
  \bauthor{\bsnm{Geraghty},~\bfnm{Sheela}\binits{S.}},
  \bauthor{\bsnm{Lanphear},~\bfnm{Bruce}\binits{B.}},
  \bauthor{\bsnm{Khoury},~\bfnm{Jane}\binits{J.}},
  \bauthor{\bsnm{Rey},~\bfnm{Javier}\binits{J.}},
  \bauthor{\bsnm{Dewitt},~\bfnm{Thomas}\binits{T.}} \AND
  \bauthor{\bsnm{Britto},~\bfnm{Maria}\binits{M.}}
(\byear{2007}).
\btitle{Effect of Multisource Feedback on Resident Communication Skills and
  Professionalism}.
\bjournal{Archives of pediatrics \& adolescent medicine}
\bvolume{161}
\bpages{44-9}.
\end{barticle}
\endbibitem

\bibitem[\protect\citeauthoryear{Cai, Campbell and Broderick}{2021}]{Cai:2021}
\begin{binproceedings}[author]
\bauthor{\bsnm{Cai},~\bfnm{Diana}\binits{D.}},
  \bauthor{\bsnm{Campbell},~\bfnm{Trevor}\binits{T.}} \AND
  \bauthor{\bsnm{Broderick},~\bfnm{Tamara}\binits{T.}}
(\byear{2021}).
\btitle{Finite mixture models do not reliably learn the number of components}.
In \bbooktitle{Proceedings of the 38th International Conference on Machine
  Learning}
(\beditor{\bfnm{Marina}\binits{M.}~\bsnm{Meila}} \AND
  \beditor{\bfnm{Tong}\binits{T.}~\bsnm{Zhang}}, eds.).
\bseries{Proceedings of Machine Learning Research}
\bvolume{139}
\bpages{1158--1169}.
\bpublisher{PMLR}.
\end{binproceedings}
\endbibitem

\bibitem[\protect\citeauthoryear{Chen and Gopalakrishnan}{1998}]{Chen:1998}
\begin{binproceedings}[author]
\bauthor{\bsnm{Chen},~\bfnm{Scott~Shaobing}\binits{S.~S.}} \AND
  \bauthor{\bsnm{Gopalakrishnan},~\bfnm{Ponani~S}\binits{P.~S.}}
(\byear{1998}).
\btitle{Clustering via the Bayesian information criterion with applications in
  speech recognition}.
In \bbooktitle{Proceedings of the 1998 IEEE International Conference on
  Acoustics, Speech and Signal Processing, ICASSP'98 (Cat. No. 98CH36181)}
\bvolume{2}
\bpages{645--648}.
\bpublisher{IEEE}.
\end{binproceedings}
\endbibitem

\bibitem[\protect\citeauthoryear{{Ferreira da Silva}}{2007}]{Adelino:2007}
\begin{barticle}[author]
\bauthor{\bsnm{{Ferreira da Silva}},~\bfnm{Adelino~R.}\binits{A.~R.}}
(\byear{2007}).
\btitle{A Dirichlet process mixture model for brain MRI tissue classification}.
\bjournal{Medical Image Analysis}
\bvolume{11}
\bpages{169-182}.
\end{barticle}
\endbibitem

\bibitem[\protect\citeauthoryear{de~Souto et~al.}{2008}]{Souto:2008}
\begin{barticle}[author]
\bauthor{\bparticle{de} \bsnm{Souto},~\bfnm{Marc{\'i}lio
  Carlos~Pereira}\binits{M.~C.~P.}},
  \bauthor{\bsnm{Costa},~\bfnm{Ivan~G.}\binits{I.~G.}}, \bauthor{\bparticle{de}
  \bsnm{Araujo},~\bfnm{Daniel S.~A.}\binits{D.~S.~A.}},
  \bauthor{\bsnm{Ludermir},~\bfnm{Teresa~Bernarda}\binits{T.~B.}} \AND
  \bauthor{\bsnm{Schliep},~\bfnm{Alexander}\binits{A.}}
(\byear{2008}).
\btitle{Clustering cancer gene expression data: a comparative study}.
\bjournal{BMC Bioinformatics}
\bvolume{9}
\bpages{497 - 497}.
\end{barticle}
\endbibitem

\bibitem[\protect\citeauthoryear{Devroye and Wagner}{1977}]{Devroye:1977}
\begin{barticle}[author]
\bauthor{\bsnm{Devroye},~\bfnm{Luc~P.}\binits{L.~P.}} \AND
  \bauthor{\bsnm{Wagner},~\bfnm{T.~J.}\binits{T.~J.}}
(\byear{1977}).
\btitle{{The Strong Uniform Consistency of Nearest Neighbor Density
  Estimates}}.
\bjournal{The Annals of Statistics}
\bvolume{5}
\bpages{536 -- 540}.
\bdoi{10.1214/aos/1176343851}
\end{barticle}
\endbibitem

\bibitem[\protect\citeauthoryear{Donoho}{2000}]{Donoho:2000}
\begin{barticle}[author]
\bauthor{\bsnm{Donoho},~\bfnm{David}\binits{D.}}
(\byear{2000}).
\btitle{High-Dimensional Data Analysis: The Curses and Blessings of
  Dimensionality}.
\bjournal{AMS Math Challenges Lecture}
\bpages{1-32}.
\end{barticle}
\endbibitem

\bibitem[\protect\citeauthoryear{Forero, Kekatos and
  Giannakis}{2011}]{Forero:2011}
\begin{barticle}[author]
\bauthor{\bsnm{Forero},~\bfnm{Pedro}\binits{P.}},
  \bauthor{\bsnm{Kekatos},~\bfnm{Vassilis}\binits{V.}} \AND
  \bauthor{\bsnm{Giannakis},~\bfnm{G.~B.}\binits{G.~B.}}
(\byear{2011}).
\btitle{Robust Clustering Using Outlier-Sparsity Regularization}.
\bjournal{Computing Research Repository - CORR}
\bvolume{60}.
\bdoi{10.1109/TSP.2012.2196696}
\end{barticle}
\endbibitem

\bibitem[\protect\citeauthoryear{Fr{\"u}hwirth-Schnatter}{2006}]{Fruhwurth:2006}
\begin{bbook}[author]
\bauthor{\bsnm{Fr{\"u}hwirth-Schnatter},~\bfnm{S.}\binits{S.}}
(\byear{2006}).
\btitle{Finite Mixture and Markov Switching Models}.
\bseries{Springer Series in Statistics}.
\bpublisher{Springer New York}.
\end{bbook}
\endbibitem

\bibitem[\protect\citeauthoryear{Gibbs and Su}{2002}]{Gibbs:2002}
\begin{barticle}[author]
\bauthor{\bsnm{Gibbs},~\bfnm{Alison~L.}\binits{A.~L.}} \AND
  \bauthor{\bsnm{Su},~\bfnm{Francis~Edward}\binits{F.~E.}}
(\byear{2002}).
\btitle{On Choosing and Bounding Probability Metrics}.
\bjournal{International Statistical Review}
\bvolume{70}
\bpages{419--435}.
\end{barticle}
\endbibitem

\bibitem[\protect\citeauthoryear{Gorsky, Chan and Ma}{2020}]{Gorsky:2020}
\begin{barticle}[author]
\bauthor{\bsnm{Gorsky},~\bfnm{Shai}\binits{S.}},
  \bauthor{\bsnm{Chan},~\bfnm{Cliburn}\binits{C.}} \AND
  \bauthor{\bsnm{Ma},~\bfnm{Li}\binits{L.}}
(\byear{2020}).
\btitle{Coarsened mixtures of hierarchical skew normal kernels for flow
  cytometry analyses}.
\bjournal{arXiv}.
\end{barticle}
\endbibitem

\bibitem[\protect\citeauthoryear{Greenspan, Ruf and
  Goldberger}{2006}]{Greenspan:2006}
\begin{barticle}[author]
\bauthor{\bsnm{Greenspan},~\bfnm{H.}\binits{H.}},
  \bauthor{\bsnm{Ruf},~\bfnm{A.}\binits{A.}} \AND
  \bauthor{\bsnm{Goldberger},~\bfnm{J.}\binits{J.}}
(\byear{2006}).
\btitle{Constrained Gaussian mixture model framework for automatic segmentation
  of MR brain images}.
\bjournal{IEEE Transactions on Medical Imaging}
\bvolume{25}
\bpages{1233-1245}.
\end{barticle}
\endbibitem

\bibitem[\protect\citeauthoryear{Gretton et~al.}{2012}]{Gretton:2012}
\begin{barticle}[author]
\bauthor{\bsnm{Gretton},~\bfnm{Arthur}\binits{A.}},
  \bauthor{\bsnm{Borgwardt},~\bfnm{Karsten~M.}\binits{K.~M.}},
  \bauthor{\bsnm{Rasch},~\bfnm{Malte~J.}\binits{M.~J.}},
  \bauthor{\bsnm{Sch{{\"o}}lkopf},~\bfnm{Bernhard}\binits{B.}} \AND
  \bauthor{\bsnm{Smola},~\bfnm{Alexander}\binits{A.}}
(\byear{2012}).
\btitle{A Kernel Two-Sample Test}.
\bjournal{Journal of Machine Learning Research}
\bvolume{13}
\bpages{723-773}.
\end{barticle}
\endbibitem

\bibitem[\protect\citeauthoryear{Guha, Ho and Nguyen}{2021}]{Guha:2021}
\begin{barticle}[author]
\bauthor{\bsnm{Guha},~\bfnm{Aritra}\binits{A.}},
  \bauthor{\bsnm{Ho},~\bfnm{Nhat}\binits{N.}} \AND
  \bauthor{\bsnm{Nguyen},~\bfnm{XuanLong}\binits{X.}}
(\byear{2021}).
\btitle{{On posterior contraction of parameters and interpretability in
  Bayesian mixture modeling}}.
\bjournal{Bernoulli}
\bvolume{27}
\bpages{2159 -- 2188}.
\bdoi{10.3150/20-BEJ1275}
\end{barticle}
\endbibitem

\bibitem[\protect\citeauthoryear{Joyce}{2011}]{Joyce:2011}
\begin{bincollection}[author]
\bauthor{\bsnm{Joyce},~\bfnm{James~M}\binits{J.~M.}}
(\byear{2011}).
\btitle{Kullback-leibler divergence}.
In \bbooktitle{International encyclopedia of statistical science}
\bpages{720--722}.
\bpublisher{Springer}.
\end{bincollection}
\endbibitem

\bibitem[\protect\citeauthoryear{Miller and Dunson}{2019}]{Miller:2019}
\begin{barticle}[author]
\bauthor{\bsnm{Miller},~\bfnm{Jeffrey~W.}\binits{J.~W.}} \AND
  \bauthor{\bsnm{Dunson},~\bfnm{David~B.}\binits{D.~B.}}
(\byear{2019}).
\btitle{Robust Bayesian Inference via Coarsening}.
\bjournal{Journal of the American Statistical Association}
\bvolume{114}
\bpages{1113-1125}.
\bdoi{10.1080/01621459.2018.1469995}
\end{barticle}
\endbibitem

\bibitem[\protect\citeauthoryear{Paninski}{2003}]{Paninski:2003}
\begin{barticle}[author]
\bauthor{\bsnm{Paninski},~\bfnm{Liam}\binits{L.}}
(\byear{2003}).
\btitle{Estimation of Entropy and Mutual Information}.
\bjournal{Neural Computation}
\bvolume{15}
\bpages{1191-1253}.
\bdoi{10.1162/089976603321780272}
\end{barticle}
\endbibitem

\bibitem[\protect\citeauthoryear{Perkel et~al.}{2019}]{Perkel:2019}
\begin{barticle}[author]
\bauthor{\bsnm{Perkel},~\bfnm{Jeffrey~M}\binits{J.~M.}} \betal{et~al.}
(\byear{2019}).
\btitle{Julia: come for the syntax, stay for the speed}.
\bjournal{Nature}
\bvolume{572}
\bpages{141--142}.
\end{barticle}
\endbibitem

\bibitem[\protect\citeauthoryear{Prabhakaran et~al.}{2016}]{Prabhakaran:2016}
\begin{binproceedings}[author]
\bauthor{\bsnm{Prabhakaran},~\bfnm{Sandhya}\binits{S.}},
  \bauthor{\bsnm{Azizi},~\bfnm{Elham}\binits{E.}},
  \bauthor{\bsnm{Carr},~\bfnm{Ambrose}\binits{A.}} \AND
  \bauthor{\bsnm{Pe’er},~\bfnm{Dana}\binits{D.}}
(\byear{2016}).
\btitle{Dirichlet Process Mixture Model for Correcting Technical Variation in
  Single-Cell Gene Expression Data}.
In \bbooktitle{Proceedings of The 33rd International Conference on Machine
  Learning}
(\beditor{\bfnm{Maria~Florina}\binits{M.~F.}~\bsnm{Balcan}} \AND
  \beditor{\bfnm{Kilian~Q.}\binits{K.~Q.}~\bsnm{Weinberger}}, eds.).
\bseries{Proceedings of Machine Learning Research}
\bvolume{48}
\bpages{1070--1079}.
\bpublisher{PMLR}, \baddress{New York, New York, USA}.
\end{binproceedings}
\endbibitem

\bibitem[\protect\citeauthoryear{Simon-Gabriel and
  Sch\"{o}lkopf}{2018}]{SimonGabriel:2018}
\begin{barticle}[author]
\bauthor{\bsnm{Simon-Gabriel},~\bfnm{Carl-Johann}\binits{C.-J.}} \AND
  \bauthor{\bsnm{Sch\"{o}lkopf},~\bfnm{Bernhard}\binits{B.}}
(\byear{2018}).
\btitle{{Kernel Distribution Embeddings - Universal Kernels, Characteristic
  Kernels and Kernel Metrics on Distributions.}}
\bjournal{Journal of Machine Learning Research}
\bvolume{19}
\bpages{1 -- 29}.
\end{barticle}
\endbibitem

\bibitem[\protect\citeauthoryear{Simon-Gabriel et~al.}{2023}]{Simon:2020}
\begin{barticle}[author]
\bauthor{\bsnm{Simon-Gabriel},~\bfnm{C.~J.}\binits{C.~J.}},
  \bauthor{\bsnm{Barp},~\bfnm{A.}\binits{A.}},
  \bauthor{\bsnm{Sch{\"o}lkopf},~\bfnm{B.}\binits{B.}} \AND
  \bauthor{\bsnm{Mackey},~\bfnm{L.}\binits{L.}}
(\byear{2023}).
\btitle{Metrizing Weak Convergence with Maximum Mean Discrepancies}.
\bjournal{Journal of Machine Learning Research}
\bvolume{24}.
\end{barticle}
\endbibitem

\bibitem[\protect\citeauthoryear{Sriperumbudur
  et~al.}{2010}]{Sriperumbudur:2010}
\begin{barticle}[author]
\bauthor{\bsnm{Sriperumbudur},~\bfnm{Bharath~K}\binits{B.~K.}},
  \bauthor{\bsnm{Gretton},~\bfnm{A.}\binits{A.}},
  \bauthor{\bsnm{Fukumizu},~\bfnm{K.}\binits{K.}},
  \bauthor{\bsnm{Sch\"{o}lkopf},~\bfnm{Bernhard}\binits{B.}} \AND
  \bauthor{\bsnm{Lanckriet},~\bfnm{G~R~G}\binits{G.~R.~G.}}
(\byear{2010}).
\btitle{{Hilbert Space Embeddings and Metrics on Probability Measures}}.
\bjournal{Journal of Machine Learning Research}
\bvolume{11}
\bpages{1517 -- 1561}.
\end{barticle}
\endbibitem

\bibitem[\protect\citeauthoryear{Stevens et~al.}{2019}]{Stevens:2019}
\begin{barticle}[author]
\bauthor{\bsnm{Stevens},~\bfnm{Elizabeth}\binits{E.}},
  \bauthor{\bsnm{Dixon},~\bfnm{Dennis~R.}\binits{D.~R.}},
  \bauthor{\bsnm{Novack},~\bfnm{Marlena~N.}\binits{M.~N.}},
  \bauthor{\bsnm{Granpeesheh},~\bfnm{Doreen}\binits{D.}},
  \bauthor{\bsnm{Smith},~\bfnm{Tristram}\binits{T.}} \AND
  \bauthor{\bsnm{Linstead},~\bfnm{Erik}\binits{E.}}
(\byear{2019}).
\btitle{Identification and analysis of behavioral phenotypes in autism spectrum
  disorder via unsupervised machine learning}.
\bjournal{International Journal of Medical Informatics}
\bvolume{129}
\bpages{29-36}.
\bdoi{https://doi.org/10.1016/j.ijmedinf.2019.05.006}
\end{barticle}
\endbibitem

\bibitem[\protect\citeauthoryear{van~der Vaart and Wellner}{1996}]{Vaart:1996}
\begin{bbook}[author]
\bauthor{\bparticle{van~der} \bsnm{Vaart},~\bfnm{AW}\binits{A.}} \AND
  \bauthor{\bsnm{Wellner},~\bfnm{J.}\binits{J.}}
(\byear{1996}).
\btitle{Weak Convergence and Empirical Processes: With Applications to
  Statistics}.
\bseries{Springer Series in Statistics}.
\bpublisher{Springer}.
\end{bbook}
\endbibitem

\bibitem[\protect\citeauthoryear{Villani}{2009}]{Villani:2009}
\begin{bbook}[author]
\bauthor{\bsnm{Villani},~\bfnm{C}\binits{C.}}
(\byear{2009}).
\btitle{{Optimal transport: old and new}}
\bvolume{338}.
\bpublisher{Springer}.
\end{bbook}
\endbibitem

\bibitem[\protect\citeauthoryear{Walker and Hjort}{2001}]{Walker:2001}
\begin{barticle}[author]
\bauthor{\bsnm{Walker},~\bfnm{Stephen}\binits{S.}} \AND
  \bauthor{\bsnm{Hjort},~\bfnm{Nils~Lid}\binits{N.~L.}}
(\byear{2001}).
\btitle{On Bayesian consistency}.
\bjournal{Journal of the Royal Statistical Society: Series B (Statistical
  Methodology)}
\bvolume{63}
\bpages{811-821}.
\bdoi{https://doi.org/10.1111/1467-9868.00314}
\end{barticle}
\endbibitem

\bibitem[\protect\citeauthoryear{Wang and Blei}{2018}]{Wang:2018}
\begin{barticle}[author]
\bauthor{\bsnm{Wang},~\bfnm{Chong}\binits{C.}} \AND
  \bauthor{\bsnm{Blei},~\bfnm{David~M.}\binits{D.~M.}}
(\byear{2018}).
\btitle{{A General Method for Robust Bayesian Modeling}}.
\bjournal{Bayesian Analysis}
\bvolume{13}
\bpages{1163 -- 1191}.
\bdoi{10.1214/17-BA1090}
\end{barticle}
\endbibitem

\bibitem[\protect\citeauthoryear{Wang and Blei}{2019}]{Wang:2019}
\begin{barticle}[author]
\bauthor{\bsnm{Wang},~\bfnm{Yixin}\binits{Y.}} \AND
  \bauthor{\bsnm{Blei},~\bfnm{David~M.}\binits{D.~M.}}
(\byear{2019}).
\btitle{Frequentist Consistency of Variational Bayes}.
\bjournal{Journal of the American Statistical Association}
\bvolume{114}
\bpages{1147-1161}.
\end{barticle}
\endbibitem

\bibitem[\protect\citeauthoryear{Wang, Kulkarni and Verdú}{2009}]{Wang:2009}
\begin{barticle}[author]
\bauthor{\bsnm{Wang},~\bfnm{Qing}\binits{Q.}},
  \bauthor{\bsnm{Kulkarni},~\bfnm{Sanjeev}\binits{S.}} \AND
  \bauthor{\bsnm{Verdú},~\bfnm{Sergio}\binits{S.}}
(\byear{2009}).
\btitle{Divergence Estimation for Multidimensional Densities Via
  -Nearest-Neighbor Distances}.
\bjournal{Information Theory, IEEE Transactions on}
\bvolume{55}
\bpages{2392 - 2405}.
\end{barticle}
\endbibitem

\bibitem[\protect\citeauthoryear{Wellner}{1981}]{Wellner:1981}
\begin{barticle}[author]
\bauthor{\bsnm{Wellner},~\bfnm{Jon~A.}\binits{J.~A.}}
(\byear{1981}).
\btitle{A Glivenko-Cantelli theorem for empirical measures of independent but
  non-identically distributed random variables}.
\bjournal{Stochastic Processes and their Applications}
\bvolume{11}
\bpages{309-312}.
\end{barticle}
\endbibitem

\bibitem[\protect\citeauthoryear{Zhao and Lai}{2020}]{Zhao:2020}
\begin{binproceedings}[author]
\bauthor{\bsnm{Zhao},~\bfnm{Puning}\binits{P.}} \AND
  \bauthor{\bsnm{Lai},~\bfnm{Lifeng}\binits{L.}}
(\byear{2020}).
\btitle{Analysis of K Nearest Neighbor KL Divergence Estimation for Continuous
  Distributions}.
In \bbooktitle{2020 IEEE International Symposium on Information Theory (ISIT)}
\bpages{2562-2567}.
\bdoi{10.1109/ISIT44484.2020.9174033}
\end{binproceedings}
\endbibitem

\bibitem[\protect\citeauthoryear{Zong et~al.}{2018}]{Zong:2018}
\begin{binproceedings}[author]
\bauthor{\bsnm{Zong},~\bfnm{Bo}\binits{B.}},
  \bauthor{\bsnm{Song},~\bfnm{Qi}\binits{Q.}},
  \bauthor{\bsnm{Min},~\bfnm{Martin~Renqiang}\binits{M.~R.}},
  \bauthor{\bsnm{Cheng},~\bfnm{Wei}\binits{W.}},
  \bauthor{\bsnm{Lumezanu},~\bfnm{Cristian}\binits{C.}},
  \bauthor{\bsnm{Cho},~\bfnm{Daeki}\binits{D.}} \AND
  \bauthor{\bsnm{Chen},~\bfnm{Haifeng}\binits{H.}}
(\byear{2018}).
\btitle{Deep Autoencoding Gaussian Mixture Model for Unsupervised Anomaly
  Detection}.
In \bbooktitle{International Conference on Learning Representations}.
\end{binproceedings}
\endbibitem

\end{thebibliography}

\newpage
\appendix

\section{Technical Lemma}

The following lemma states that integral probability metrics (as defined in Eq.\ (8)) are jointly convex -- that is, they satisfy Assumption 1(b). 
\begin{lemma}\label{lem:ipm-joint-convex}
	Suppose $P_i$ and $Q_i$, $i = 1,\ldots, n$ are probability measures defined on $\mcX$. Then for $0\leq w_i \leq 1$ and $\sum_{i = 1}^n w_i = 1$, 
	\[
	d_{\mcH}\left(\sum_{i=1}^{n}w_i P_i, \sum_{i=1}^{n}w_i Q_i\right) 
	\le \sum_{i=1}^{n}w_i  	d_{\mcH} (P_i, Q_i).
	\]
	
\end{lemma}

\begin{proof}
	By definition of the integral probability metric, 
	\begin{align}
		{d_{\mcH}\left(\sum_{i=1}^{n}w_i P_i, \sum_{i=1}^{n}w_iQ_i\right)}
		&=  \sup_{h \in \mcH}\left\vert \int_{\mcX} h(x) \left(\sum_{i=1}^{n}w_i P_i(\dee x)\right) - \int_{\mcX} h(x)\left(\sum_{i=1}^{n}w_i Q_i(\dee x)\right)\right\vert  \nonumber \\
		& = \sup_{h \in \mcH}\left\vert \sum_{i=1}^{n}w_i \left(\int_{\mcX} h(x)  P_i(\dee x) - \int_{\mcX} h(x) Q_i(\dee x) \right)\right\vert\nonumber \\
		& \le  \sup_{h \in \mcH} \sum_{i=1}^{n}w_i \left\vert  \int_{\mcX} h(x) P_i(\dee x) - \int_{\mcX} h(x) Q_i(\dee x) \right\vert \nonumber \\
		& \le \sum_{i=1}^{n}w_i  \sup_{h \in \mcH} \left\vert  \int_{\mcX} h(x) P_i(\dee x) - \int_{\mcX} h(x) Q_i(\dee x) \right\vert  \nonumber \\
		& = \sum_{i=1}^{n}w_i  d_{\mcH} (P_i, Q_i).
	\end{align}
\end{proof}

\section{Proofs}
\subsection{Notation}
We write $\op(g(\numobs))$ to denote a random function $f$ that satisfies $f(\numobs) / g(\numobs) \to 0$ in probability for $\numobs \to \infty$.
Let $\hat\param_{k}^{(\numcomps,\numobs)}$ denote the $k$th component parameter estimate using $\data{1:\numobs}$ 
for the mixture model with $\numcomps$ components. 
More generally, we replace superscript $(\numcomps)$ with $(\numcomps,\numobs)$ to make the dependence on $\numobs$ explicit. 
Let $\hat{n}^{(\numcomps,\numobs)}_k = |X^{(k)}(z^{(\numcomps,\numobs)})|$ and $\hat{\pi}^{(\numcomps,\numobs)}_{k} = \hat{n}^{(\numcomps,\numobs)}_k/\numobs$. 
Note that, with probability 1, $\hat{n}^{(\numcomps,\numobs)}_k \rightarrow \infty$ as $\numobs \rightarrow \infty$.
For simplicity, we introduce the shorthand notation 
$F^{(\numcomps,\numobs)}_{k} = F_{\hat\param^{(\numcomps,\numobs)}_k}$, 
$\mcR^{\rho}_{\numobs} (\allparam^{(\numcomps)}) = \mcR^{\rho}(\allparam^{(\numcomps)}; z^{(\numcomps,\numobs)}, \data{1:\numobs})$, and 
$\mcR^{\rho,\lambda}_{\numobs}(\allparam^{(\numcomps)}) = \mathcal{R}^{\rho,\lambda}(\allparam^{(\numcomps)}; z^{(\numcomps,\numobs)}, \data{1:\numobs})$. 

Define the conditional component probabilities based on optimal model distribution and true generating distribution respectively as
$$
p_{\star}^{(\numcomps)}(k \mid x)= \frac{\pi^{(\numcomps)}_{\star k}f_{\star k}^{(\numcomps)}(x)}{g^{(\numcomps)}_{\star}(x)}, \qquad
p_{o}(k \mid x) = \frac{\pi_{ok}p_{ok}(x)}{p_{o}(x)}.
$$
For conditional probabilities of model distribution, we denote as $\hat{p}^{(\numcomps, \numobs)}(k \mid x) = \hat{\pi}^{(\numcomps,\numobs)}_{k}f_{k}^{(\numcomps,\numobs)}(x) / g^{(\numcomps,\numobs)}(x)$.
Note that $\hat{p}^{(\numcomps,\numobs)} (k\mid x) \rightarrow p_{\star}^{(\numcomps)}(k \mid x)$ as $\numobs \rightarrow \infty$ with probability 1.

\subsection{Proof of Theorem 1} \label{pf:main-thm}

We show that (1) for $\numcomps = \numcomps_o$,  $ \mcR^{\rho}_{\numobs}(\allparam_{\star}^{(\numcomps)}) \rightarrow 0$ in probability as $\numobs \rightarrow \infty$, and (2) for $\numcomps < \numcomps_o$, $\mcR^{\rho}_{\numobs}(\allparam_{\star}^{(\numcomps)})\rightarrow \infty$ in probability as $\numobs \rightarrow \infty$.
The theorem conclusion follows immediately from these two results since the unpenalized loss is lower bounded by zero,
so the penalized loss will be asymptotically minimized at the smallest $\numcomps$ which has unpenalized loss of zero.

\textbf{Proof of part (1).} %
Fix $\numcomps=\numcomps_{o}$. 
It follows from Assumptions 1(c,d) and 2 that $\discrest{\hat{P}^{(\numcomps_{o},\numobs)}_{k}}{F^{(\numcomps,\numobs)}_{k} }  \to \discr{\tilde{P}_{k}}{F^{(\numcomps_{o})}_{\opt k}}$ in probability.
Hence, it follows from Assumption 3($\rho$)(b) that there exists $\epsilon>0$ such that 
\[
\discrest{\hat{P}^{(\numcomps_{o},\numobs)}_{k}}{F^{(\numcomps,\numobs)}_{k} } < \rho - \epsilon + \op(1).
\]
Using this inequality, it follows that  
\begin{align}
	\mcR^{\rho}_{\numobs}(\allparam_{\opt}^{(\numcomps_o)})  
	& =\sum_{k=1}^{\numcomps_o}	\hat{n}^{(\numcomps_{o},\numobs)}_k \max\{0, \discrest{\hat{P}_{k}^{(\numcomps_o, \numobs)}}{F^{(\numcomps_{o},\numobs)}_{k}}  - \rho\}\\
	& \le \sum_{k=1}^{\numcomps_o}	\hat{n}^{(\numcomps_{o},\numobs)}_k\max\{0, -\epsilon+\op(1)\}.
\end{align}
Hence, we can conclude that $\lim_{\numobs \rightarrow\infty}\Pr[ \mcR^{\rho}_{\numobs}(\allparam_{\opt}^{(\numcomps_o)}) =0] = 1$. %

\textbf{Proof of part (2).} 
Now we consider the case of $\numcomps < \numcomps_o$. 
Note the empirical distribution can be written as $\hat{P}^{(\numcomps, \numobs)} = \sum_{k=1}^{\numcomps}\hat{\pi}_k^{(\numcomps, \numobs)} \hat{P}_k^{(\numcomps, \numobs)}$.
By dominated convergence, we know that for $\numobs \to \infty$,
\begin{equation}
	\hat{\pi}_k^{(\numcomps,\numobs)} = \int \hat{p}^{(\numcomps,\numobs)} (k\mid y)P_o(\dee y) \rightarrow  \int p_{\star}^{(K)}(k\mid y)P_o(\dee y) = \pi^{(\numcomps)}_{\star k}, \label{eq:pi-conv}
\end{equation}
where convergence is in probability.

For the purpose of contradiction,
assume that for all $k = 1,\ldots, \numcomps$, $\discr{\tilde{P}_{k}}{F^{(\numcomps)}_{\opt k}} \le \rho$. 
Consider 
\begin{align}
	&d\left(G^{(\numcomps)}_{\opt}, P_{o}\right) \\
	& \le d\left(G^{(\numcomps)}_{\opt}, \hat{P}^{(\numcomps, \numobs)}\right) + d\left( \hat{P}^{(\numcomps, \numobs)}, P_{o}\right)\label{eq:trigin-ineq}\\
	& = d\left(\sum_{k = 1}^{\numcomps}  \pi^{(\numcomps)}_{\star k}F^{(\numcomps)}_{\opt k}, \sum_{k=1}^{\numcomps}\hat{\pi}_k^{(\numcomps, \numobs)} \hat{P}_k^{(\numcomps, \numobs)}\right) + \op(1) \label{eq:assump1a}\\
	&  \le d\left(\sum_{k = 1}^{\numcomps}  \pi^{(\numcomps)}_{\star k}F^{(\numcomps)}_{\opt k}, \sum_{k = 1}^{\numcomps}  \hat{\pi}_k^{(\numcomps, \numobs)} F^{(\numcomps)}_{\opt k} \right) + d\left(\sum_{k = 1}^{\numcomps}  \hat{\pi}_k^{(\numcomps, \numobs)} F^{(\numcomps)}_{\opt k}, \sum_{k=1}^{\numcomps}\hat{\pi}_k^{(\numcomps, \numobs)} \hat{P}_k^{(\numcomps, \numobs)}\right) + \op(1), \label{eq:decomp}
\end{align}
where \cref{eq:assump1a} follows by Assumption 1(a).
Define $\pi_k^{\min} = \min\{\pi^{(\numcomps)}_{\star k}, \hat{\pi}^{(\numcomps, \numobs)}_k\}$ and $\bar{\pi}=1-\sum_{k=1}^{\numcomps}\pi_k^{\min}$. 
Let $\Vert \cdot \Vert_{1}$ denote the $\ell_1$-norm.
For the first term in \cref{eq:decomp}, we can write 
\begin{align}
	&d\left(\sum_{k = 1}^{\numcomps} \pi^{(\numcomps)}_{\star k}F^{(\numcomps)}_{\opt k}, \sum_{k = 1}^{\numcomps}  \hat{\pi}_k^{(\numcomps, \numobs)} F^{(\numcomps)}_{\opt k} \right) \\
	&= d\left(\sum_{k = 1}^{\numcomps} \pi_k^{\min}F^{(\numcomps)}_{\opt k} + \bar{\pi}\sum_{k = 1}^{\numcomps}  \frac{\pi^{(\numcomps)}_{\star k}-\pi_k^{\min}}{\bar{\pi}}F^{(\numcomps)}_{\opt k}, \sum_{k = 1}^{\numcomps}  \pi_k^{\min}F^{(\numcomps)}_{\opt k}+ \bar{\pi}\sum_{k = 1}^{\numcomps}  \frac{\hat{\pi}_k^{(\numcomps, \numobs)} -  \pi_k^{\min}}{\bar{\pi}}F^{(\numcomps)}_{\opt k} \right) \\
	&\le \sum_{k = 1}^{\numcomps} \pi_k^{\min}d\left(F^{(\numcomps)}_{\opt k}, F^{(\numcomps)}_{\opt k}\right) + \bar{\pi}d\left(\sum_{k = 1}^{\numcomps}  \frac{\pi^{(\numcomps)}_{\star k}-\pi_k^{\min}}{\bar{\pi}}F^{(\numcomps)}_{\opt k}, \sum_{k = 1}^{\numcomps}  \frac{\hat{\pi}_k^{(\numcomps, \numobs)} -  \pi_k^{\min}}{\bar{\pi}}F^{(\numcomps)}_{\opt k}\right)\label{eq:joint-convexity}\\
	&\le \Vert \pi^{(\numcomps)}_{\star} - \hat{\pi}^{(\numcomps,\numobs)} \Vert_{1} d\left(\sum_{k = 1}^{\numcomps}  \frac{\pi^{(\numcomps)}_{\star k}-\pi_k^{\min}}{\bar{\pi}}F^{(\numcomps)}_{\opt k}, \sum_{k = 1}^{\numcomps}  \frac{\hat{\pi}_k^{(\numcomps, \numobs)} -  \pi_k^{\min}}{\bar{\pi}}F^{(\numcomps)}_{\opt k}\right)\label{eq:pi-l1norm-bound-pibar}\\
	&= \Vert \pi^{(\numcomps)}_{\star} - \hat{\pi}^{(\numcomps,\numobs)}\Vert_{1} \\
	& \qquad \times  d\left( \sum_{k = 1}^{\numcomps} \sum_{l = 1}^{\numcomps} \frac{(\pi^{(\numcomps)}_{\star k}-\pi_k^{\min})(\hat{\pi}_l^{(\numcomps, \numobs)} -  \pi_l^{\min}) }{\bar{\pi}^2}  F^{(\numcomps)}_{\opt k}, \sum_{k = 1}^{\numcomps} \sum_{l = 1}^{\numcomps} \frac{(\pi^{(\numcomps)}_{\star k}-\pi_k^{\min})(\hat{\pi}_l^{(\numcomps, \numobs)} -  \pi_l^{\min}) }{\bar{\pi}^2} F^{(\numcomps)}_{\opt l}\right)\\ 
	&\le \Vert \pi^{(\numcomps)}_{\star} - \hat{\pi}^{(\numcomps,\numobs)}\Vert_{1} \sum_{k = 1}^{\numcomps} \sum_{l = 1}^{\numcomps} \frac{(\pi^{(\numcomps)}_{\star k}-\pi_k^{\min})(\hat{\pi}_l^{(\numcomps, \numobs)} -  \pi_l^{\min}) }{\bar{\pi}^2} d\left(F^{(\numcomps)}_{\opt k}, F^{(\numcomps)}_{\opt l}\right) \\
	&\le \Vert \pi^{(\numcomps)}_{\star} - \hat{\pi}^{(\numcomps,\numobs)}\Vert_{1} \max_{\substack{k,l\in \{1,\ldots,K\} \\ k\neq l}} d\left(F^{(\numcomps)}_{\opt k}, F^{(\numcomps)}_{\opt l}\right) \\
	& = \op(1), \label{eq:first-bound}
\end{align}
where \cref{eq:joint-convexity} uses Assumption 1(b), \cref{eq:pi-l1norm-bound-pibar} follows by the fact that $\bar{\pi}=1- \sum_{k=1}^{\numcomps}\pi_k^{\min} \le  \sum_{k=1}^{\numcomps}\pi_k^{\max}-\sum_{k=1}^{\numcomps}\pi_k^{\min}=\Vert \pi^{(\numcomps)}_{\star} - \hat{\pi}^{(\numcomps,\numobs)} \Vert_{1}$, and \cref{eq:first-bound} follows by Assumption 1(f) and \cref{eq:pi-conv}.

For the second term in \cref{eq:decomp}, we can upper bounded as
\begin{align}
	d\left(\sum_{k = 1}^{\numcomps}  \hat{\pi}_k^{(\numcomps, \numobs)} F^{(\numcomps)}_{\opt k}, \sum_{k=1}^{\numcomps}\hat{\pi}_k^{(\numcomps, \numobs)} \hat{P}_k^{(\numcomps, \numobs)}\right) 
	&\le   \sum_{k = 1}^{\numcomps}  \hat{\pi}_k^{(\numcomps, \numobs)}d\left( F^{(\numcomps)}_{\opt k},  \hat{P}_k^{(\numcomps, \numobs)}\right)  \label{eq:assump1c}\\
	&\le   \sum_{k = 1}^{\numcomps}  \hat{\pi}_k^{(\numcomps, \numobs)}\phi\left(\discr{\hat{P}_k^{(\numcomps, \numobs)}}{F^{(\numcomps)}_{\opt k}}\right)  \label{eq:assump1d}\\
	& \le  \sum_{k = 1}^{\numcomps}  \hat{\pi}_k^{(\numcomps, \numobs)}\phi\left(\rho\right) \label{eq:prop}\\
	&= \phi(\rho)\label{eq:second-bound},  
\end{align}
where %
\cref{eq:assump1c} follows by Assumption 1(b), 
\cref{eq:assump1d} follows by Assumption 1(e), and 
\cref{eq:prop} follows by our assumption for purposes of contradiction.

Plugging \cref{eq:first-bound,eq:second-bound} into \cref{eq:decomp} yields the final inequality $d\left(G^{(\numcomps)}_{\opt}, P_{o}\right) \le \phi(\rho) + \op(1)$, which contradicts Assumption 3($\rho$)(c).
Therefore, there must exist $\tilde{k}$ such that $\discr{\tilde{P}_{\tilde{k}}}{F^{(\numcomps)}_{\opt\tilde{k}}} > \rho$.
Hence, for some $\epsilon > 0$, $\discr{\hat{P}_{\tilde{k}}^{(\numcomps,\numobs)} }{F^{(\numcomps,\numobs)}_{\tilde{k}}} = \rho + \epsilon +\op(1)$.
Hence, we have
\begin{align}
	\mcR^{\rho}_{\numobs}(\allparam_{\opt}^{(\numcomps)}) 
	&\ge \hat{n}_{\tilde{k}}^{(\numcomps,\numobs)}\max\{0, \discrest{\hat{P}_{\tilde{k}}^{(\numcomps,\numobs)} }{F^{(\numcomps,\numobs)}_{\tilde{k}}} - \rho\} \\
	&= \hat{n}_{\tilde{k}}^{(\numcomps,\numobs)}  \max\{0, \discr{\hat{P}_{\tilde{k}}^{(\numcomps,\numobs)} }{F^{(\numcomps,\numobs)}_{\tilde{k}}} - \rho +\op(1)\}  \label{ineq:assump1c}\\
	&= \hat{n}_{\tilde{k}}^{(\numcomps,\numobs)}\max\{0, \epsilon + \op(1)\} \\
	&\rightarrow \infty,\label{ineq:N-infty}
\end{align} 
where \cref{ineq:assump1c} follows from Assumption 1(c) and \cref{ineq:N-infty} follows
since $\hat{n}_{\tilde{k}}^{(\numcomps,\numobs)}  \rightarrow \infty$ in probability for $\numobs\rightarrow\infty$.

\subsection{Proof of Proposition 1}

Assumption 1(a) follows by \citet[Theorem 1]{Wellner:1981}.
Assumption 1(b) holds by \cref{lem:ipm-joint-convex}.
Assumption 1(e) holds with $\phi(\rho) =(\rho/2)^{1/2}$ by the fact that, letting $d_{\mathrm{TV}}$ 
denote total variation distance, $d_{\mathrm{BL}} \le d_{\mathrm{TV}}$ and 
$(2d_{\mathrm{TV}})^{2}  \le \mathrm{KL}$ \citep[][Section 3]{Gibbs:2002}.
Assumption 1(f) holds since $d_{\mathrm{BL}} \le 1 < \infty$.

\subsection{Proof of Proposition 2}
Assumption 1(a) follows by the assumption that $d_{\mathrm{MMD}}$ metrizes weak convergence.
Assumption 1(b) holds by \cref{lem:ipm-joint-convex}.
Assumption 1(e) holds by choosing $\phi(\rho) =\rho$ for maximum mean discrepancy. Assumption 1(f) holds for maximum mean discrepancy with bounded kernels since $d_{\mathrm{MMD}} \le \sup_{x} \mcK(x,x)^{1/2} < \infty$.

\subsection{Proof of Proposition 3}

Let $p_{ok}$ be the density for the true $k$th component distribution. Let $g^{(K)}_{\star}(x) =\sum_{\ell=1}^{\numcomps}\pi^{(\numcomps)}_{\star\ell}f^{(\numcomps)}_{\star\ell}(x)$. 
It follows by the definition of $\Pr_{\opt}(z=k\mid x)$ and Assumption 3($\rho$)(a) that
\begin{align}
	\tilde{P}_{k}(\dee x) &= \frac{p_{\star}^{(\numcomps_{o})}(k \mid x)P_o(\dee x) }{\int p_{\star}^{(\numcomps_{o})}(k \mid y)P_o(\dee y)} \label{appx-eq:phat-def}  = f^{(\numcomps_{o})}_{\opt k}(x) \frac{\sum_{\ell=1}^{\numcomps_{o}}\pi_{ok}P_{o\ell}(\dee x)}{\sum_{\ell=1}^{\numcomps_{o}}\pi^{(\numcomps_{o})}_{\opt\ell} f^{(\numcomps_{o})}_{\opt\ell}(x)}.
\end{align}
We can rewrite $\tilde{P}_{k}(\dee x) $ in terms of the rations $p_{\star}^{(\numcomps_{o})}(k \mid x)/p_{o}(k \mid x)$ and $\pi_{ok}/\pi^{(\numcomps_{o})}_{\star k}$:
\begin{align}
	\tilde{P}_{k}(\dee x)  &=  \frac{\pi_{ok}}{\pi^{(\numcomps_{o})}_{\star k}}\cdot  \frac{\pi^{(\numcomps_{o})}_{\star k}f_{\star k}^{(\numcomps_{o})}(x)}{\pi_{ok}p_{ok}(x)}\cdot \frac{\sum_{\ell=1}^{\numcomps_{o}}\pi_{ok}p_{o\ell}(x)}{\sum_{\ell=1}^{\numcomps_{o}}\pi^{(\numcomps_{o})}_{\star \ell}f_{\star \ell}^{(\numcomps_{o})}(x)}\cdot P_{ok}(\dee x).\\
	&= \frac{\pi_{ok}}{\pi^{(\numcomps_{o})}_{\star k}} \cdot \frac{\pi^{(\numcomps_{o})}_{\star k}f_{\star k}^{(\numcomps_{o})}(x)}{\sum_{\ell=1}^{\numcomps_{o}}\pi^{(\numcomps_{o})}_{\star \ell}f_{\star \ell}^{(\numcomps_{o})}(x)} \cdot \frac{\sum_{\ell=1}^{\numcomps_{o}}\pi_{ok}p_{o\ell}(x)}{\pi_{ok}p_{ok}(x)}\cdot P_{ok}(\dee x)\\
	& = \frac{\pi_{ok}}{\pi^{(\numcomps_{o})}_{\star k}}\cdot \frac{p_{\star}^{(\numcomps_{o})}(k \mid x)}{p_{o}(k \mid x)}\cdot P_{ok}(\dee x) \label{eq:rewrite-phat-k}
\end{align}
To bound $\kl{\tilde{P}_{k}}{F^{(\numcomps_{o})}_{\opt k}}$, we plug in the expression for $	\tilde{P}_{k}(\dee x)$ in \cref{eq:rewrite-phat-k} and get
\begin{align}
	\kl{\tilde{P}_{k}}{F^{(\numcomps_{o})}_{\opt k} } &= \int \tilde{P}_{k}(\dee x)\log\frac{d \tilde{P}_{k}}{d F_{\phi^{\opt}_{k}}}(x) \\
	& = \int \frac{\pi_{ok}}{\pi^{(\numcomps_{o})}_{\star k}}\cdot \frac{p_{\star}^{(\numcomps_{o})}(k \mid x)}{p_{o}(k \mid x)} \\ 
	& \qquad \cdot P_{ok}(\dee x) \bigg\{\log\frac{\pi_{ok}}{\pi^{(\numcomps_{o})}_{\star k}} + \log\frac{p_{\star}^{(\numcomps_{o})}(k \mid x)}{p_{o}(k \mid x)} + \log\frac{dP_{ok}}{dF^{(\numcomps_{o})}_{\opt k}}(x)    \bigg\} \nonumber\\
	& \le (1+\epsilon_{\pi})(1+\epsilon_{z})\bigg\{\int \log \left(1+\epsilon_{\pi})(1+\epsilon_{z}\right) P_{ok}(\dee x)  \label{ineq:use-ration-assump1} \\ 
	& \phantom{\le (1+\epsilon_{\pi})(1+\epsilon_{z})\bigg\{~}  + \int P_{ok}(\dee x) \log\frac{dP_{ok}}{dF^{(\numcomps_{o})}_{\opt k}}(x)     \bigg\} \nonumber \\
	& < (1+\epsilon_{\pi})(1+\epsilon_{z})\left\{\log (1+\epsilon_{\pi})(1+\epsilon_{z}) + \kl{P_{ok}}{F^{(\numcomps_{o})}_{\opt k}}  \right\}  \label{ineq:use-ration-assump2} \\
	& < (1+\epsilon_{\pi})(1+\epsilon_{z})\left\{\log (1+\epsilon_{\pi})(1+\epsilon_{z}) + \rho_o\right\} \label{ineq:use-kl-assump},
\end{align}
where \cref{ineq:use-ration-assump1,ineq:use-ration-assump2} follow by applying the assumptions on the weight ratio and posterior probability ratio and \cref{ineq:use-kl-assump} follows by the upper bound on $\kl{P_{ok}}{F^{(\numcomps_{o})}_{\opt k}}$. Hence, we may set $\rho = (1+\epsilon_{\pi})(1+\epsilon_{z})\left[\log (1+\epsilon_{\pi})(1+\epsilon_{z}) + \rho_o \right]$.

\subsection{Proof of Proposition 4}
It follows by the triangle inequality, \cref{eq:rewrite-phat-k}, H\"older's inequality, and the assumption that all $h \in \mcH$ are bounded by $M$ that
\begin{align}
	&d_{\mcH}(\tilde{P}_{k},F^{(\numcomps_{o})}_{\opt k})\\
	& \le d_{\mcH}(\tilde{P}_{k},P_{ok}) + d_{\mcH}(P_{ok},F^{(\numcomps_{o})}_{\opt k}) \\
	& \le d_{\mcH}\left(\frac{\pi_{ok}}{\pi^{(\numcomps_{o})}_{\star k}} \frac{p_{\star}^{(\numcomps_{o})}(k \mid x)}{p_{o}(k \mid x)} P_{ok},P_{ok}\right) + \rho_{o}\\
	& = \sup_{h\in \mcH} \left\vert \int h(x) \frac{\pi_{ok}}{\pi^{(\numcomps_{o})}_{\star k}} \frac{p_{\star}^{(\numcomps_{o})}(k \mid x)}{p_{o}(k \mid x)} P_{ok}(\dee x) - \int h(y) P_{ok}(\dee y)   \right\vert
	+ \rho_{o} \\
	& =  \sup_{h\in \mcH} \left\vert \int \left(\frac{\pi_{ok}}{\pi^{(\numcomps_{o})}_{\star k}} \frac{p_{\star}^{(\numcomps_{o})}(k \mid x)}{p_{o}(k \mid x)}-1\right) h(x) P_{ok}(\dee x)   \right\vert
	+ \rho_{o} \\
	& \le \int P_{ok}(\dee x)\cdot  \sup_{h\in \mcH} \sup_{x \in \mcX} \left\vert  h(x)   \right\vert \left\vert \frac{\pi_{ok}}{\pi^{(\numcomps_{o})}_{\star k}} \frac{p_{\star}^{(\numcomps_{o})}(k \mid x)}{p_{o}(k \mid x)}-1 \right\vert  + \rho_{o} \\ 
	& \le M \left(\epsilon_{\pi}+\epsilon_{z}+\epsilon_{\pi}\epsilon_{z}\right)
	+ \rho_{o}.
\end{align}
Hence, we may take $\rho = M\left(\epsilon_{\pi}+\epsilon_{z}+\epsilon_{\pi}\epsilon_{z}\right)  +\rho_{o}$.

\section{Connection to Likelihood-based Inference}

For the Kullback--Leibler divergence case, we can relate the structurally aware loss to the conditional negative log-likelihood.
Let $\tilde{X}_{k}$ be a random element selected uniformly from $X^{(k)}(z_{1:\numobs})$. 
Then the conditional negative log-likelihood given $z_{1:\numobs}$ is 
\begin{align}
	-\log p(x_{1:\numobs} \mid \theta, z_{1:\numobs}) 
	&= \sum_{k=1}^{\numcomps}\sum_{x \in X^{(k)}(z_{1:\numobs})} - \log f_{\param_{k}}(x) \\
	&= \sum_{k=1}^{\numcomps} |X^{(k)}(z_{1:\numobs})| E\{-\log f_{\param_{k}}(\tilde X_{k})\}
\end{align}
For the sake of argument, if $\tilde{X}_{k}$ were distributed according to a density $\tilde{p}_{k}$, then
the Kullback--Leibler divergence between $\tilde{p}_{k}$ and $f_{\param_{k}}$ would be  
\[
\begin{aligned}
	\kl{\tilde{p}_{k}}{ f_{\param_{k}}}
	&= E\{\log \tilde{p}_{k}(\tilde X_{k})\} -  E\{\log f_{\param_{k}}(\tilde X_{k})\}  \\
	&= -\mathcal{H}({\tilde{p}_{k}}) -  E\{\log f_{\param_{k}}(\tilde X_{k})\},
\end{aligned}
\]
where $\mathcal{H}(p) = \int p(x) \log p(x) \dee x$ denotes the entropy of a density $p$. 
Then the negative conditional log-likelihood for each is equal to the Kullback--Leibler divergence, up to an entropy term that depends on the data (i.e., $\tdp_{k}$) but not the parameter $\param_{k}$:
\[
-\log p(x_{1:\numobs} \mid \theta, z_{1:\numobs}) 
\approx \sum_{k=1}^{\numcomps}|X^{(k)}(z_{1:\numobs})| \left\{ \kl{\tilde{p}_{k}}{ f_{\param_{k}}} + \mathcal{H}({\tilde{p}_{k}})\right\}.
\]
Thus, we can view the structurally aware loss as targeting the negative conditional log-likelihood but  
(a) using a consistent estimator $\klest{X^{(k)}(z_{1:\numobs})}{f_{\param_{k}}}$ in place of $\kl{\tilde p_{k}}{f_{\param_{k}}}$ and 
(b) ``coarsening'' the Kullback--Leiber divergence using the map $t \mapsto \max(0, t - \rho)$ to avoid overfitting. 

\section{Kullback-Leibler Divergence Estimation} \label{appx:kl-estimator}

\subsection{Theory and Methods}

Following \citet{Wang:2009}, we derive and study the theory of various one-sample Kullback-Leibler estimators on continuous distributions. Estimating Kullback-Leibler between continuous distributions is nontrival. 
One common way is to start with density estimations.

Consider a general case on $\mcX = R^{D}$.
Suppose $y_{1:\numobs} = (y_1, \ldots, y_{\numobs}) \in \mcX^{\otimes \numobs}$ are independent, 
identically distributed from a continuous distribution $P$ with density $p$. 
For $r>0$, one can estimate the density $p(y_n)$ by
\begin{equation}
	P(V_{D}(r)) \approx p(y_{n}) V_{D}(r),
	\label{eq:knn-density-intuition}
\end{equation}
where $V_{D}(r) = \pi^{D/2}r^D/\Gamma(D/2+1)$ is the volume of a $D$-dimensional ball centered at $y_{n}$ of radius $r$. 
Fix the query point $y_{n}$. 
The radius $r$ can be determined by finding the $k$-th nearest neighbor $y_{n(k)}$ of $y_{n}$, i.e., $r_{k,n} = \Vert y_{n(k)} - y_n\Vert$, where $\Vert \cdot \Vert$ denotes the Euclidean distance. 
Therefore, the ball centered at $y_{n}$ with radius $r_{k,n}$ contains $k$ points and thus $P(V_{D}(r_{k,n}))$ can be estimated by $k/(\numobs-1)$. 
Plugging this estimate back to \cref{eq:knn-density-intuition} yields the $k$-nearest-neighbor density estimator for $p(y_n)$,
\begin{equation}
	\hat{p}_{\numobs}(y_n) = \frac{k/(\numobs-1)}{V_{D}(r_{k,n})}.
	\label{eq:knn-density-est-p}
\end{equation}

To estimate Kullback-Leibler divergence, \citet{Wang:2009} studied various two-sample estimators given two sets of samples $y_{1},\dots, y_{N} \sim P$ and $z_{1},\dots,z_{M} \sim Q$ where the distributions $P$ and $Q$ are unknown.
However, in the context of our method, we want to estimate the Kullback-Leibler divergence with one set of samples $y_{1:\numobs}$ and one known distribution from our assumed model $Q$. 
Hence, we modify the two-sample $k$-nearest-neighbor estimators from \citep{Wang:2009} to 
create one-sample Kullback--Leibler estimators. 

Given samples $y_{1},\dots, y_{N} \sim P$, where $P$ is unknown, and a known distribution $Q$ with density $q$, 
we can use \cref{eq:knn-density-est-p} to obtain the one-sample $k$-nearest-neighbor estimator
\begin{equation}
	\begin{aligned}
		\klestsub{b}{k}{y_{1:\numobs}}{Q}
		&= \frac{1}{\numobs}\sum_{n=1}^{\numobs}\log\left\{\frac{\hat{p}_{\numobs}(y_n)}{q(y_n)}\right\}\\
		& = \frac{1}{N}\sum_{n=1}^{\numobs} \log\left\{\frac{k/(\numobs-1)}{V_{D}(r_{k,n}) q(y_n)}\right\}.
		\label{eq:canonical-knn-kl-est}
	\end{aligned}
\end{equation}
Following the proof of \citet[Theorem 1]{Wang:2009}, we can show that for fixed $k$
\begin{equation}
	\lim\limits_{n\rightarrow\infty}E[\klestsub{b}{k}{y_{1:\numobs}}{Q}] = \kl{P}{Q} +\log k- \psi(k),
	\label{eq:consistency-stare-knn-kl-est}
\end{equation}
where $\psi(k) = \Gamma'(k)/\Gamma(k)$ is the digamma function. \cref{eq:consistency-stare-knn-kl-est} suggests that this canonical Kullback--Leibler estimator is asymptotically biased. 
However, using \cref{eq:consistency-stare-knn-kl-est}, we can define the consistent (asymptotically unbiased) estimator 
\begin{equation}
	\klestsub{u}{k}{y_{1:\numobs}}{Q} = \klestsub{b}{k}{y_{1:\numobs}}{Q} - \log k + \psi(k).
	\label{eq:biased-correct-knn-kl-est}
\end{equation}
Another way to eliminate the bias is to make $k$ data-dependent, which we call \emph{adaptive} $k$-nearest-neighbor estimators. 
Following the proof of \citet[Theorem 5]{Wang:2009}, we can show that $\klestsub{b}{k_{\numobs}}{y_{1:\numobs}}{Q}$ is asymptotically consistent by choosing $k_{N}$ to satisfy mild growth conditions. 
\begin{proposition} \label{prop:one-sample-klest}
	Suppose $P$ and $Q$ are distributions uniformly continuous on $\mathbb{R}^D$ with densities $p$ and $q$, 
	and $\kl{P}{Q} < \infty$. Let $k_{\numobs}$ be a positive integer satisfying 
	\[
	\frac{k_{\numobs}}{\numobs} \rightarrow 0, \qquad \frac{k_{\numobs}}{\log \numobs} \rightarrow \infty.
	\]
	If $\inf_{p(y)>0} p(y)>0$ and $\inf_{q(y)>0} q(y)>0$, then
	\begin{equation}
		\lim\limits_{n\rightarrow\infty}\klestsub{b}{k_{\numobs}}{y_{1:\numobs}}{Q} = \kl{P}{Q}
	\end{equation}
	almost surely.
\end{proposition}
\begin{proof}
	Let $p$ and $q$ are densities of $P$ and $Q$ respectively.
	Consider the following decomposition of the error
	\begin{equation}
		\begin{aligned} 
			& \left|\klestsub{b}{k_{\numobs}}{y_{1:\numobs}}{Q} -\kl{P}{Q} \right| \\ 
			& \leq\left|\frac{1}{\numobs} \sum_{n=1}^{\numobs} \log\left\{\frac{\hat{p}_{\numobs}\left(y_n\right)}{q\left(y_n\right)}\right\} -\frac{1}{\numobs} \sum_{n=1}^{\numobs} \log\left\{\frac{p\left(y_n\right)}{q\left(y_n\right)}\right\} \right| +\left|\frac{1}{\numobs} \sum_{n=1}^{\numobs} \log\left\{\frac{p\left(y_n\right)}{q\left(y_n\right)}\right\} -\kl{P}{Q}\right| \\ 
			& \leq \frac{1}{\numobs} \sum_{n=1}^{\numobs}\left|\log \hat{p}_{\numobs}\left(y_n\right)-\log p\left(y_n\right)\right|  +\left|\frac{1}{\numobs} \sum_{n=1}^{\numobs} \log\left\{\frac{p\left(y_n\right)}{q\left(y_n\right)}\right\} -\kl{P}{Q}\right| \\ &= e_1+e_2 .\end{aligned}
	\end{equation}
	It follows by the conditions that $k_{\numobs}/\numobs \rightarrow 0$ and $k_{\numobs}/\log \numobs \rightarrow \infty$ 	and the theorem given in \cite{Devroye:1977} that $\hat{p}_{\numobs}$ is uniformly strongly consistent: almost surely
	\begin{equation}
		\lim _{\numobs \rightarrow \infty} \sup _y\left|\hat{p}_{\numobs}(y)-p(y)\right| \rightarrow 0.
	\end{equation}
	Therefore, following the proof of \cite{Wang:2009}, for any $\epsilon > 0$, there exists $N_1$ such that for any $n > N_1$, $e_1 < \epsilon/2$. For $e_2$, it simply follows by the Law of Large Numbers that for any $\epsilon > 0$, there exists $N_2 $ such that for any $n > N_2$, $e_2 < \epsilon/2$. By choosing $N = \max(N_1, N_2)$, for any $n>N$, we have $|\klestsub{b}{k_{\numobs}}{y_{1:\numobs}}{Q} -\kl{P}{Q} | < \epsilon$.
\end{proof}

\subsection{Empirical Comparison}

We now empirically compare the behavior of these $k$-nearest-neighbor Kullback--Leibler estimators.
Consider two multivariate Gaussian distributions $P = \distNorm(\mu_1, \Sigma_1)$ and $Q = \distNorm(\mu_2, \Sigma_2)$. The theoretical value for the Kullback--Leibler divergence between $P$ and $Q$ is
\begin{equation}
	\kl{P}{Q} = \frac{1}{2}\left[\log\frac{|\Sigma_2|}{|\Sigma_1|} - d + tr(\Sigma_2^{-1}\Sigma_1) + (\mu_2-\mu_1)^T\Sigma_2^{-1}(\mu_2-\mu_1)\right],
	\label{eq:kl-gaussian}
\end{equation}
where $|\cdot|$ is the determinant of a matrix and $tr(\cdot)$ denotes the trace. 
We generate samples $y_{1:\numobs}$ from $P$ and estimate $\klest{y_{1:\numobs}}{Q}$ with the three estimators above:
the canonical fixed $k$ estimator in \cref{eq:canonical-knn-kl-est} with $k\in\{1,10\}$,  the bias-corrected estimator in \cref{eq:biased-correct-knn-kl-est} with $k\in\{1,10\}$ and the adaptive estimator with $k_{\numobs} =  \numobs^{1/2}$.

We generate samples from a weakly correlated multivariate Gaussian distirbution. Set $P = \distNorm(\mu, \Sigma)$ with $\mu=(1,\ldots,1)\in R^D$ and $\Sigma_{ij}=\exp\{-(i-j)^2/\sigma^2\}$, where large $\sigma$ results in high correlations and vice versa. 
Let $\sigma=0.6$ and set $Q = \distNorm(0, I_{D})$.
We test the performance of each estimator with varying $\numobs \in \{100, 1000, 5000, 10000, 20000, 50000\}$ and varying dimensions $D \in \{4, 10, 25, 50\}$.

As shown in \cref{fig:knn-kl-est-comparison}, when $D=4$, the adaptive estimator with $k_{\numobs}=\numobs^{1/2}$ outperforms and shows reliable estimation when sample size is large ($\numobs \ge 5000$). 
This scenario resembles the setup in our simulation and real-data experiments.
We therefore use the adaptive estimator with $k_{\numobs}=\numobs^{1/2}$ for our experiments in Section 6.

When the dimension increases, the stability of all $k$-nearest-neighbor estimators drops due to the sparsity of data in high dimensions.
This reveals a limitation of all $k$-nearest-neighbor estimators.
Although proposing estimators for divergence is beyond the scope of the paper,
we test one possible adaption in Section 6.2 to use the $k$-nearest-neighbor estimators in high dimensions by assuming independence across coordinates.

\begin{figure}[hp]
	\centering	
	\subfloat[$d = 4$]{\includegraphics[width=65mm]{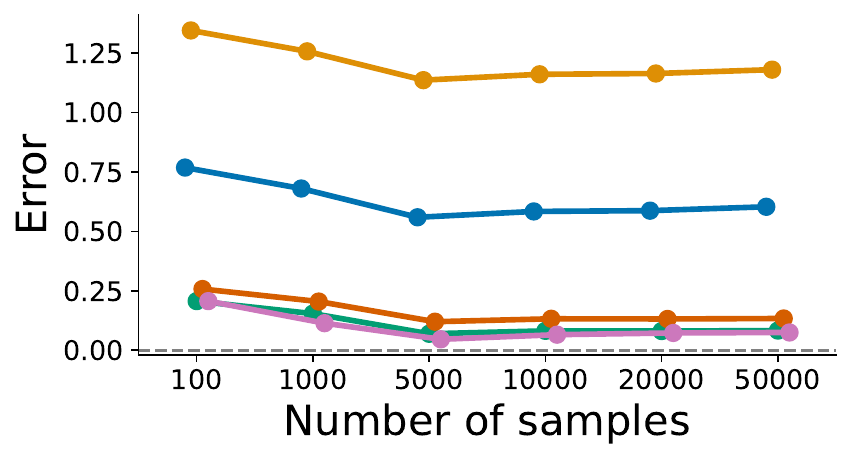}}
	\subfloat[$d = 10$]{\includegraphics[width=65mm]{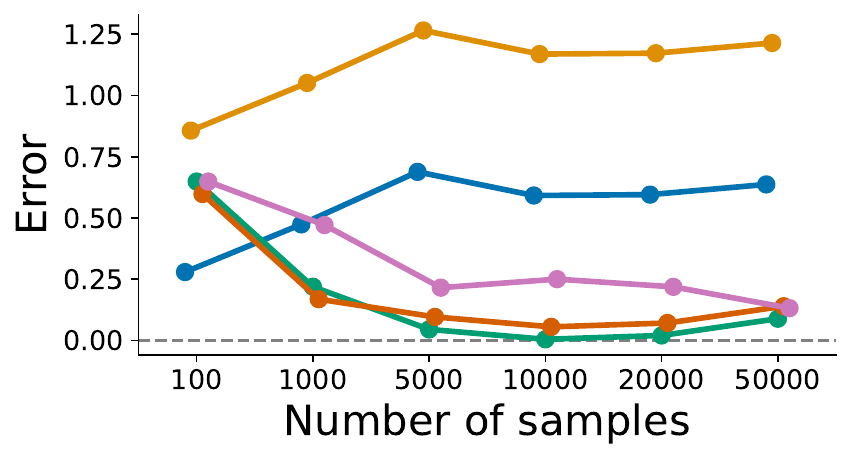}}
	\\
	\subfloat[$d = 25$]{\includegraphics[width=65mm]{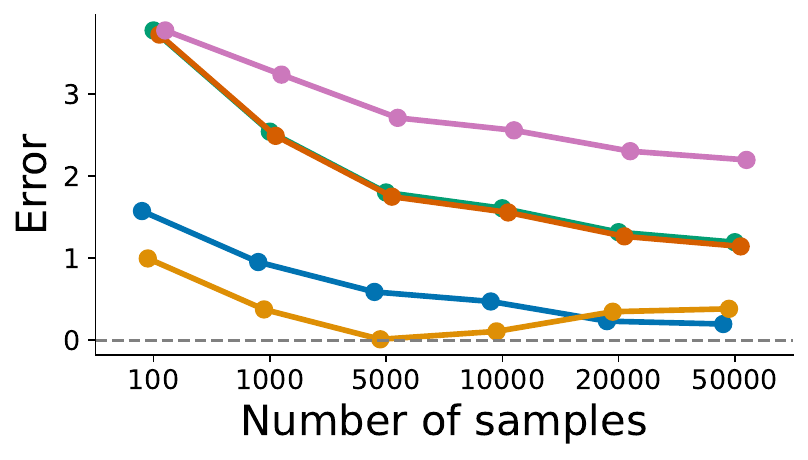}}
	\subfloat[$d = 50$]{\includegraphics[width=65mm]{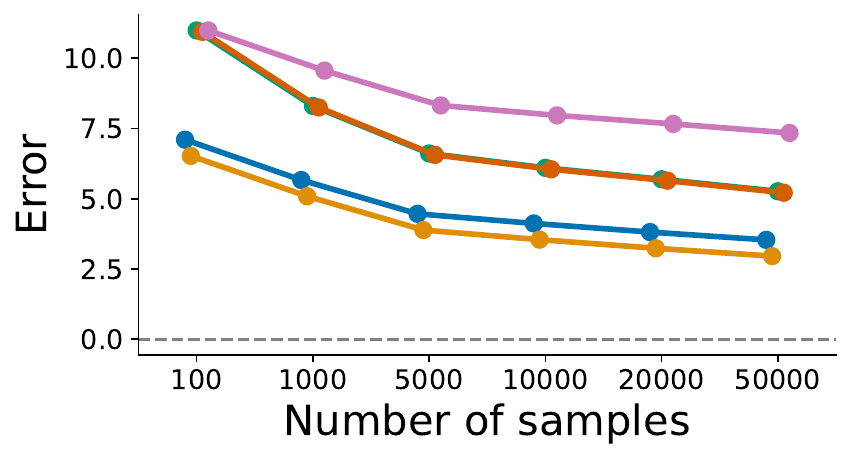}}
	\caption{Absolute error against sample size for canonical $1$-nearest-neighbor estimator (blue), canonical $10$-nearest-neighbor estimator (green), bias-corrected $1$-nearest-neighbor estimator (yellow), bias-corrected $10$-nearest-neighbor estimator (orange) and adaptive $k_{\numobs}$-nearest-neighbor estimator with $k_{\numobs}=\numobs^{1/2}$ (pink). Each panel correpsonds with a different dimension $D\in\{4,10,25,50\}$. Gray dotted lines indicate no error.} 
	\label{fig:knn-kl-est-comparison}
\end{figure}

\section{Additional Calibration Figures}

\subsection{Simulation Study}
\label{appx:simulation-gauss}

The coarsened posterior requires calibration of the hyperparameter $\alpha$, which determines the degree of misspecification.
For all scenarios considered in Section 2 and Section 6.1, we select $\alpha$ using the \emph{elbow method} proposed by \citet{Miller:2019}.
In this section, we include all calibration figures for the coarsened posterior following the code provided by \citet{Miller:2019}.

As shown in \cref{fig:coarsen-calibration}, we calibrate $\alpha$ as the turning points where we see no significant increase
in the log-likelihood if $\alpha$ increases.
With these elbow values for $\alpha$, we can see for all cases except the \texttt{small-large} case, the coarsened posterior consistently estimates the number of clusters (even after removing mini clusters with size $<2\%$) as $\hat{\numcomps} = 3 > \numcomps_{o} = 2$. 

\begin{figure}[tp]
	\centering	
	\includegraphics[width=65mm]{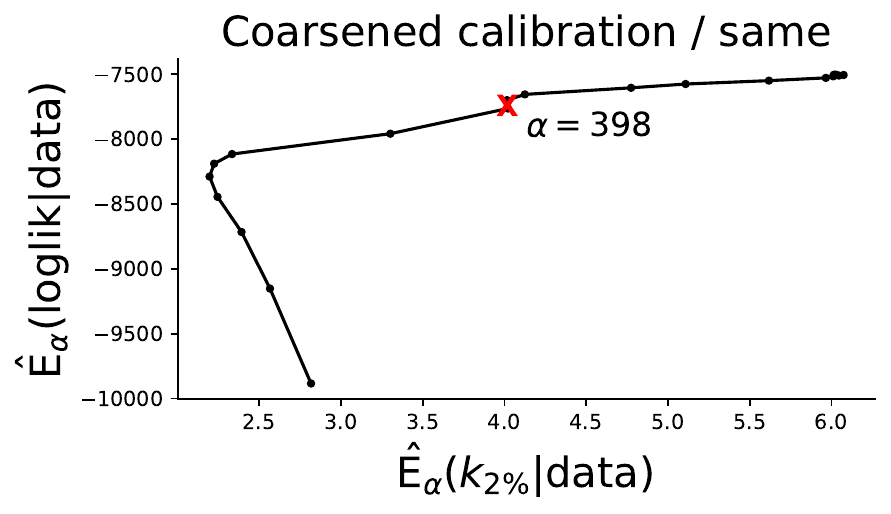}
	\includegraphics[width=65mm]{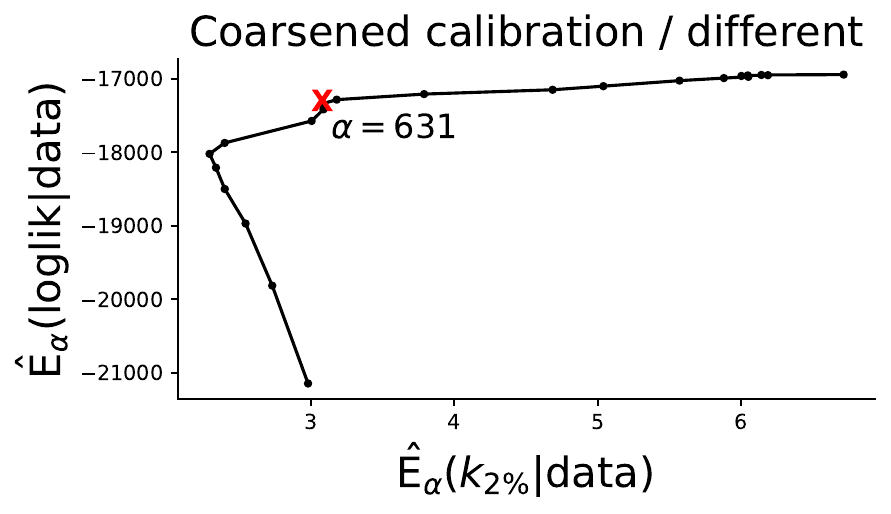}\\
	\includegraphics[width=65mm]{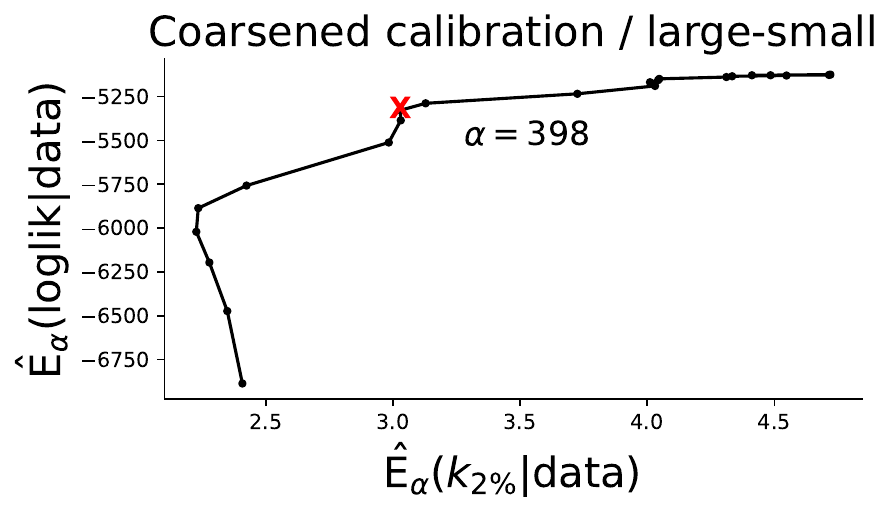}
	\includegraphics[width=65mm]{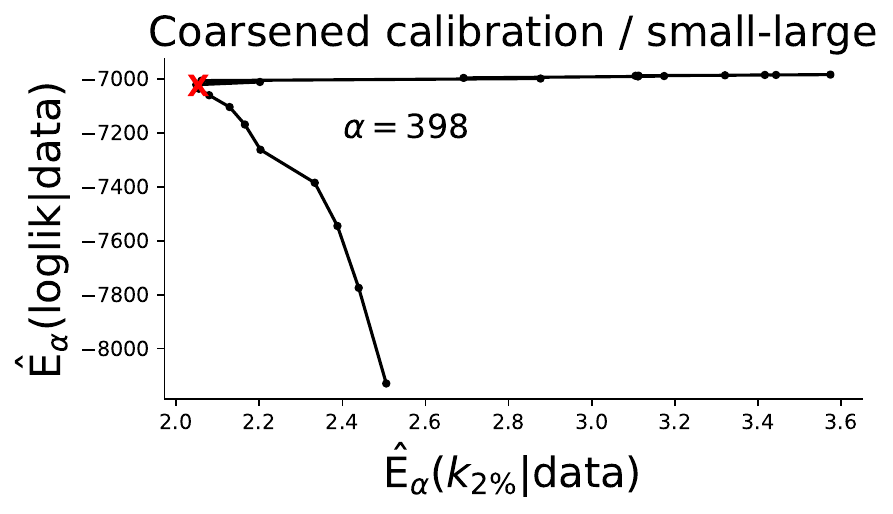}\\
	\includegraphics[width=65mm]{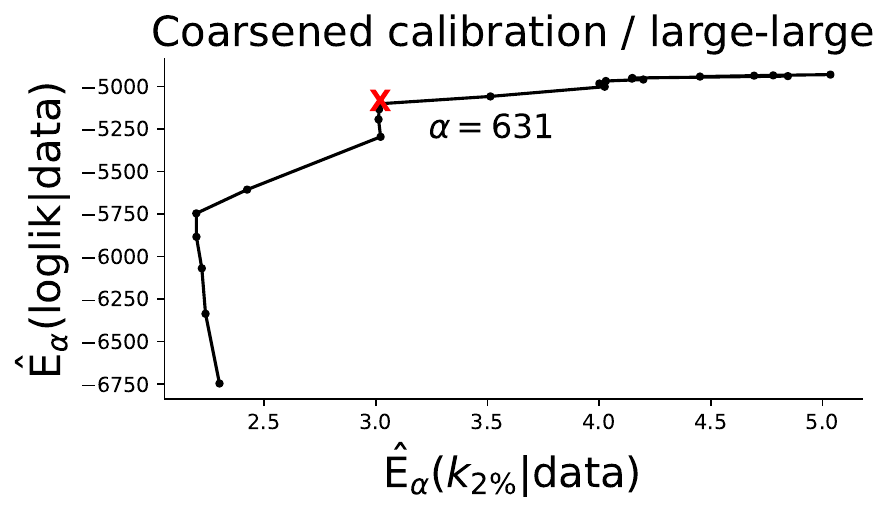}
	\caption{For the mixture of skew-normals example from Section 2 and Section 6.1, each panel shows the expected logliklihood $\hat{E}_{\alpha}(\mathrm{loglik}\mid \mathrm{data})$ against the expected number of clusters which excludes tiny clusters of size less that $2\%$ of whole dataset denoted as $\hat{E}_{\alpha}(k_{2\%}\mid \mathrm{data})$. 
		We select $\alpha$ as the elbow point in the plots.} 
	\label{fig:coarsen-calibration}
\end{figure}

\subsection{Flow Cytometry Data}
\label{appx:flow-cytometry}

In this section, we include loss and F-measure plots of our model selection method on all test datasets 7--12. 
See \citet[Section 5.2]{Miller:2019} for a discussion of the exact calibration procedure for the coarsened posterior.

Recall that to calibrate $\rho$, we select $\rho$ that optimizes the F-measure across first $6$ datasets.
To incorporate this prior knowledge on test datasets, we suggests selecting the value of $\numcomps$ 
that has has stable penalized loss and is closest to the optimal $\rho$. 
We compare our selection $\hat{\numcomps}$ with the ground truth $\numcomps_o$ labeled by experts.
For each dataset, there is always one cluster labeled as unknown due to some unclear information for cells.
With automatic clustering algorithm, it is natural for the algorithm to identify those unlabeled points and assign them to other clusters, which results in $\numcomps_o-1$ clusters.
So we treat both $\numcomps_o$ and $\numcomps_o-1$ as ground truth in our analysis. 
As shown in \cref{fig:GvHD,fig:GvHD2}, our selection method results in highest F-measure for datasets 8--12. 
Dataset 7 is challenging and even the ground truth does not produce a large F-measure.

\begin{figure}[tp]
	\centering	
	\subfloat[Data 7]{\label{fig:data7}
		\includegraphics[width=65mm]{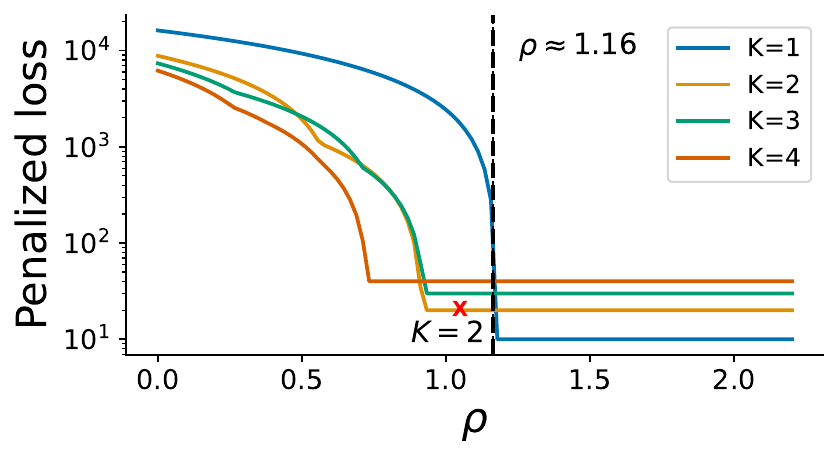}
		\includegraphics[width=65mm]{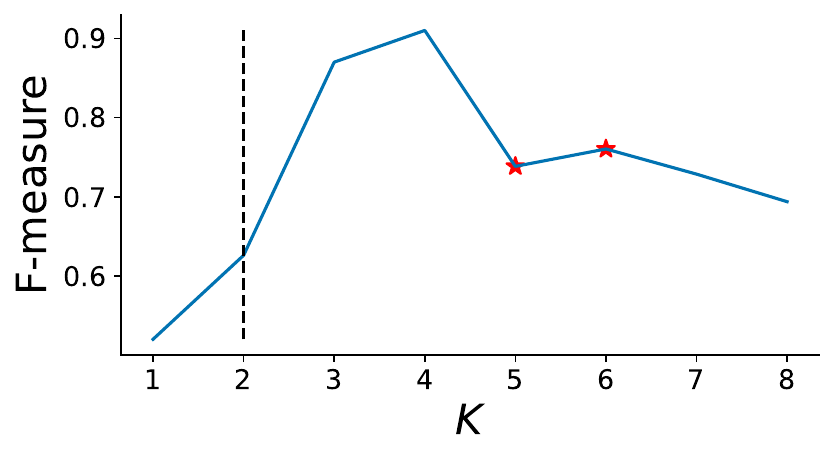}}	\\
	\subfloat[Data 8]{\label{fig:data8}
		\includegraphics[width=65mm]{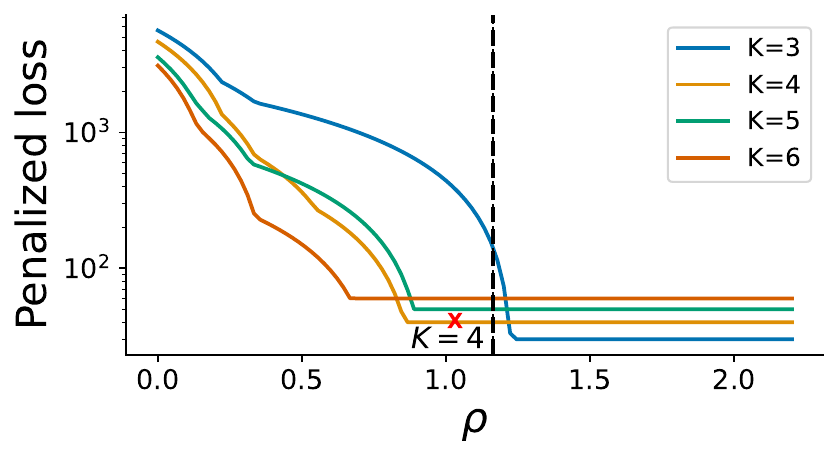}
		\includegraphics[width=65mm]{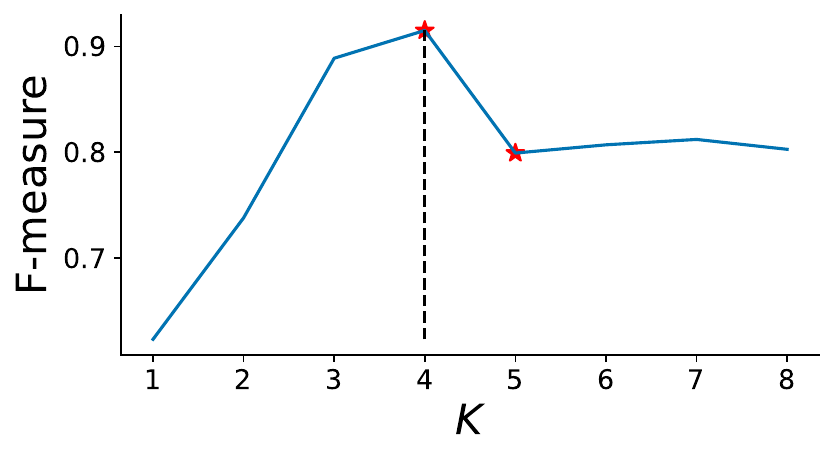}}	\\
	\subfloat[Data 9]{\label{fig:data9}
		\includegraphics[width=65mm]{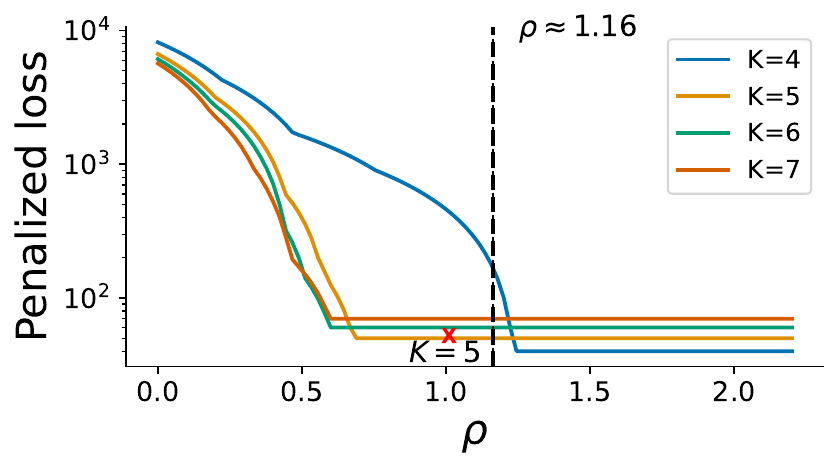}
		\includegraphics[width=65mm]{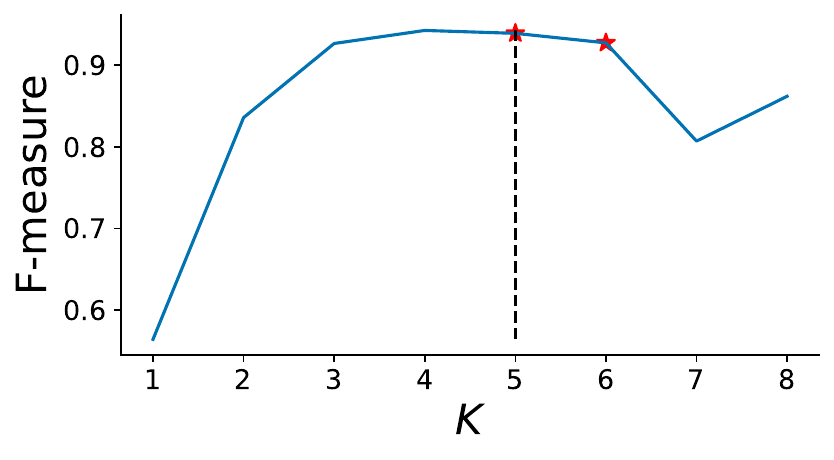}}
	\caption{Calibration and F-measure plots for test datasets 7--9 in flow cytometry experiments. \textbf{Left}: The black dashed lines indicate the optimal $\rho$ calibrated on training datasets 1--6. The cross mark indicates the selection for number of clusters. \textbf{Right}: F-measure against the number of clusters. The dashed line shows the number of clusters selected by our method and the red star indicates the ground truth $\numcomps_{o}$.} 
	\label{fig:GvHD}
\end{figure}

\begin{figure}[tp]
	\centering	
	\subfloat[Data 10]{\label{fig:data10}
		\includegraphics[width=65mm]{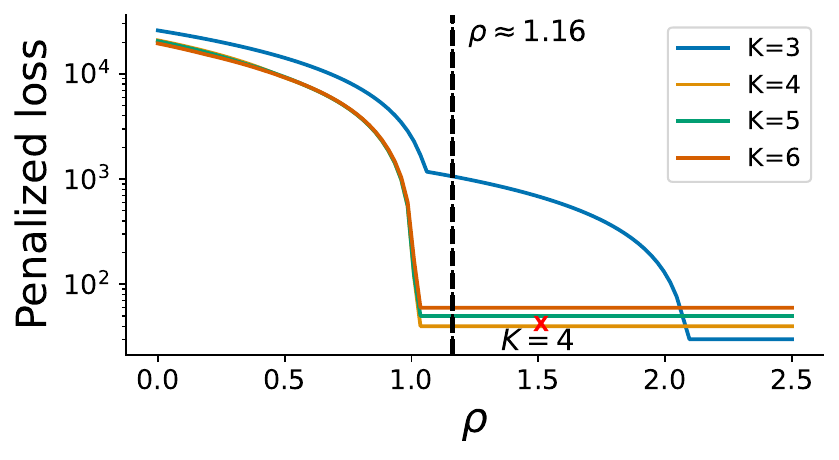}
		\includegraphics[width=65mm]{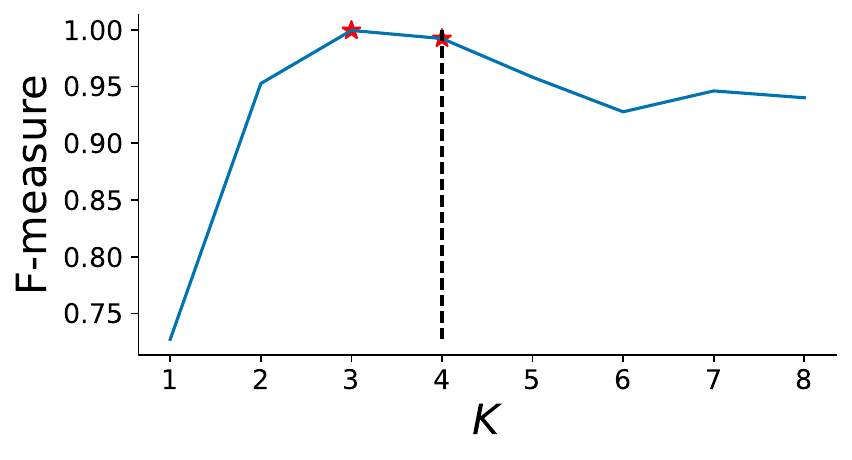}}	\\
	\subfloat[Data 11]{\label{fig:data11}
		\includegraphics[width=65mm]{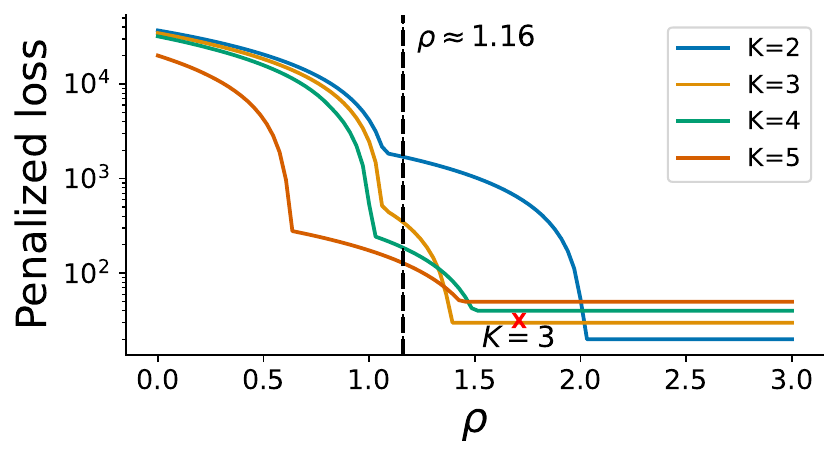}
		\includegraphics[width=65mm]{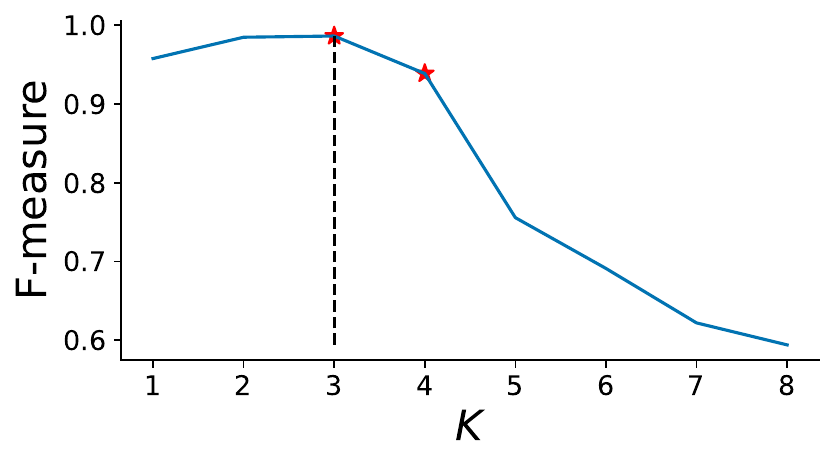}}	\\
	\subfloat[Data 12]{\label{fig:data12}
		\includegraphics[width=65mm]{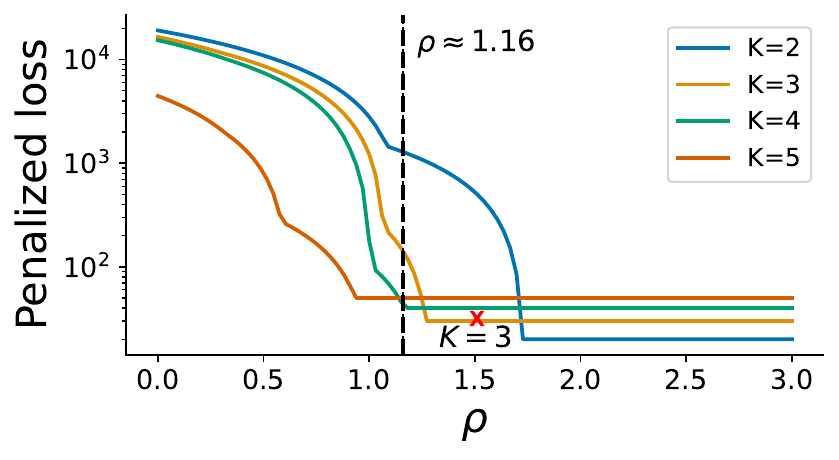}
		\includegraphics[width=65mm]{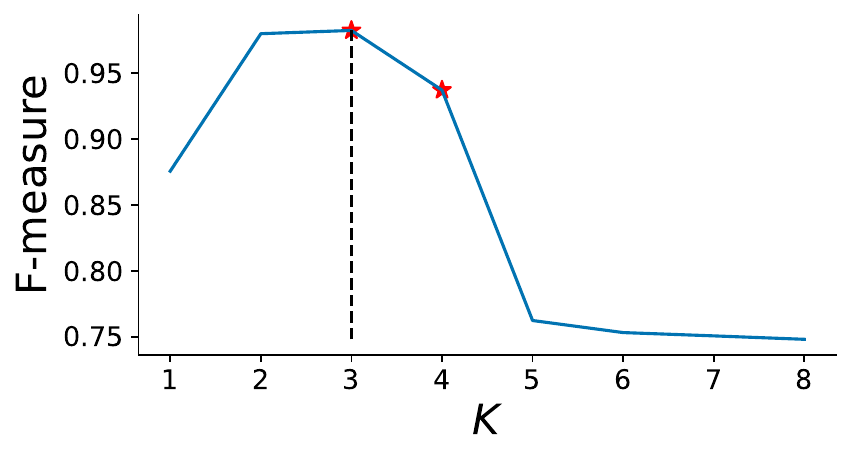}}	
	\caption{Calibration and F-measure plots for test datasets $10-12$ in flow cytometry experiments. See caption for \cref{fig:GvHD} for details}  
	\label{fig:GvHD2}
\end{figure}

\begin{figure}[tp]
	\centering	
	\subfloat[Data 7, $\numcomps=8$]{\label{fig:data7-loss}
		\includegraphics[width=65mm]{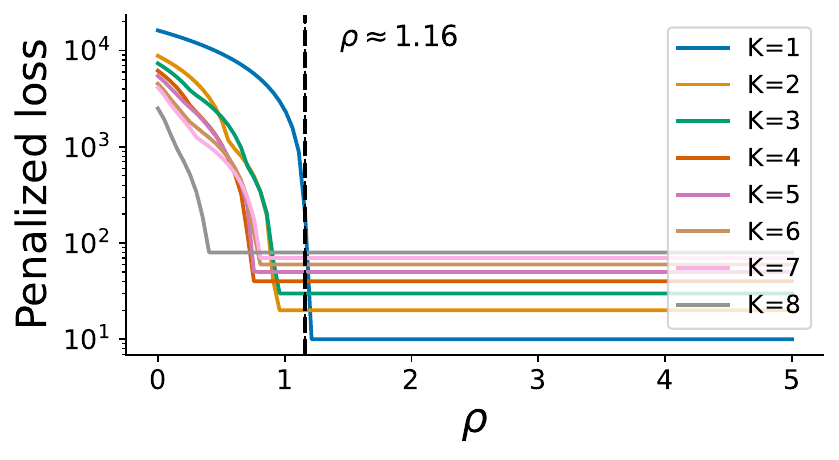}}
	\subfloat[Data 8, $\numcomps=7$]{\label{fig:data8-loss}
		\includegraphics[width=65mm]{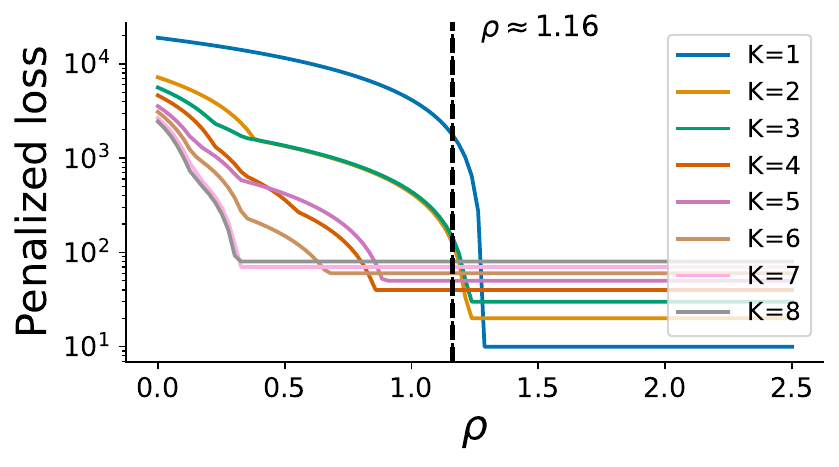}}\\
	\subfloat[Data 9, $\numcomps=5$]{\label{fig:data9-loss}
		\includegraphics[width=65mm]{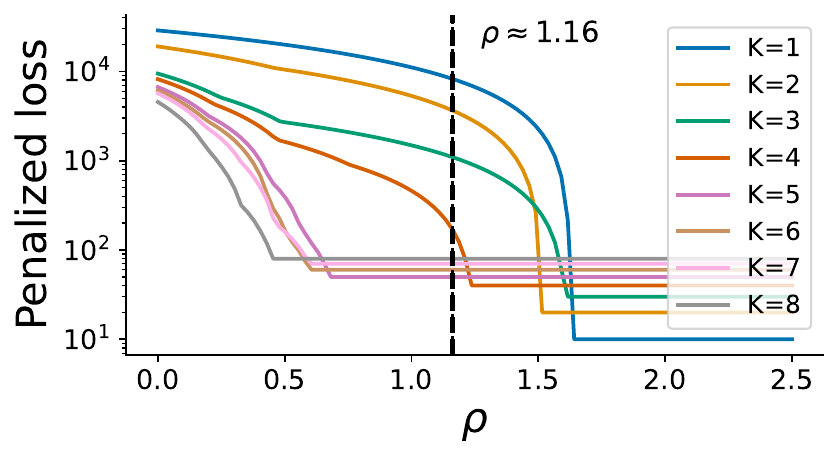}}
	\subfloat[Data 10, $\numcomps=4$]{\label{fig:data10-loss}
		\includegraphics[width=65mm]{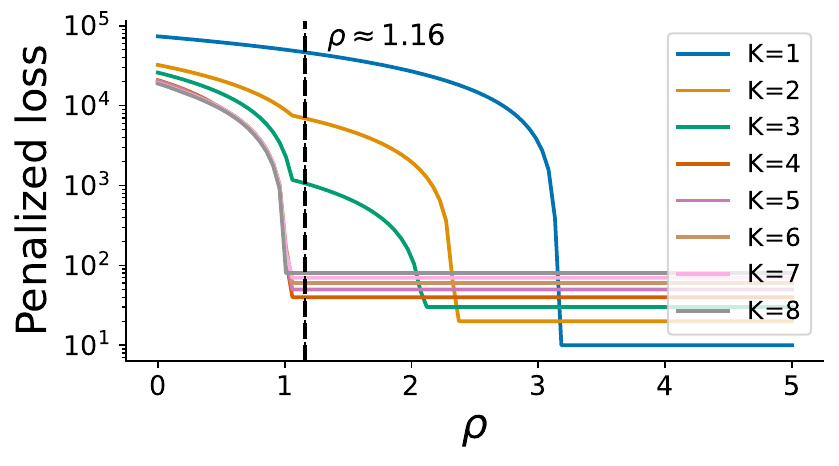}}\\
	\subfloat[Data 11, $\numcomps=3$]{\label{fig:data11-loss}
		\includegraphics[width=65mm]{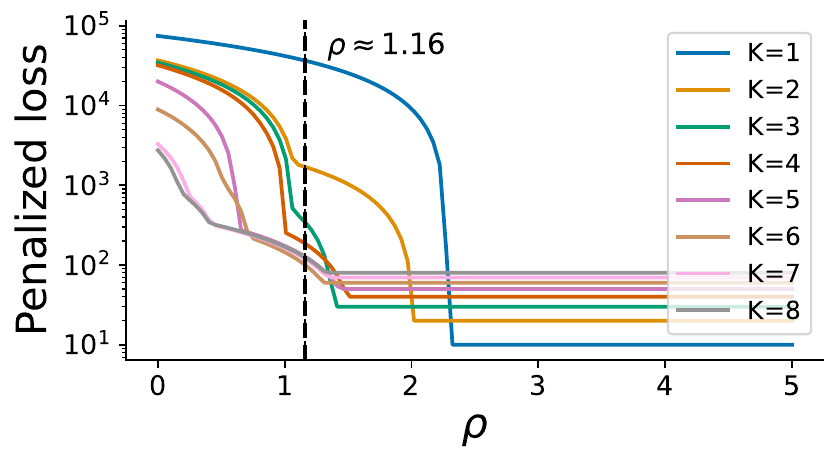}}
	\subfloat[Data 12, $\numcomps=3$]{\label{fig:data12-loss}
		\includegraphics[width=65mm]{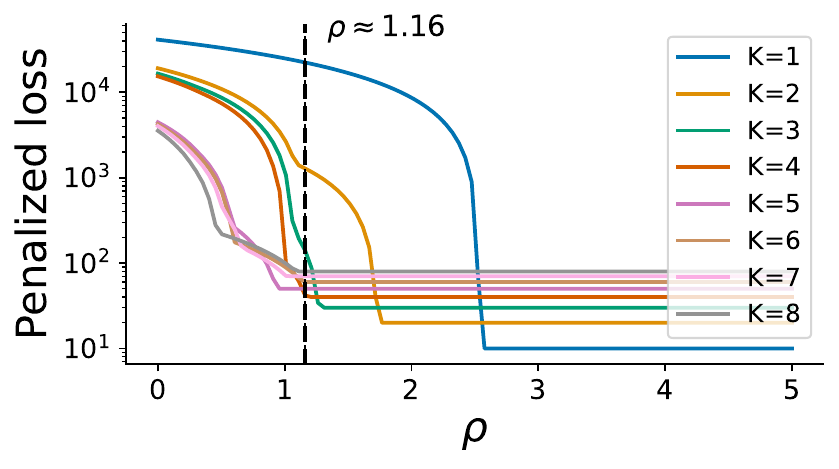}}
	\caption{Calibration including $\numcomps=1,\ldots,8$ for test datasets 7--12 in flow cytometry experiments. }  
	\label{fig:GvHD3}
\end{figure}

\end{document}